\documentclass[12pt,a4paper]{article}
\usepackage{appendix}
\usepackage{multirow}
\usepackage{float}
\usepackage[centertags]{amsmath}
\usepackage{amsfonts,amsthm,amssymb}
\usepackage{amssymb}
\usepackage{amsmath}
\usepackage{url}
\usepackage{placeins}
\usepackage{graphicx}
\usepackage{accents}
\usepackage{booktabs}
\usepackage{appendix}
\usepackage[hmargin=3.175cm,vmargin=3.175cm]{geometry}
\usepackage{etoolbox}
\usepackage{xparse}
\usepackage{enumitem}
\usepackage{authblk}
\usepackage{color}
\usepackage{chngcntr}
\usepackage{authblk}
\newcommand{\ubar}[1]{\underaccent{\bar}{#1}}
\DeclareDocumentCommand{\publicBelief}{O{\mu}}{#1}
\DeclareDocumentCommand{\privateBelief}{O{p}}{#1}
\DeclareDocumentCommand{\signalupdate}{O{q}}{#1}
\DeclareDocumentCommand{\signal}{O{s}}{#1}
\DeclareDocumentCommand{\signalsSet}{O{S}}{#1}
\DeclareDocumentCommand{\state}{O{\omega}}{#1}
\DeclareDocumentCommand{\statesSet}{O{\Omega}}{#1}
\DeclareDocumentCommand{\price}{O{\tau}}{#1}
\DeclareDocumentCommand{\action}{O{a}}{#1}
\DeclareDocumentCommand{\limitParam}{O{\alpha}}{#1}
\DeclareDocumentCommand{\deterrencePrice}{O{\price} O{d}}{#1^#2}
\DeclareDocumentCommand{\LLR}{O{x}}{\log(\frac{#1}{1-#1})}
\DeclareDocumentCommand{\llr}{O{x}}{l\left(#1\right)}
\DeclareDocumentCommand{\lBound}{O{\limitParam} O{\publicBelief}}{\ubar{#1}_{#2}}
\DeclareDocumentCommand{\uBound}{O{\limitParam} O{\publicBelief}}{\bar{#1}_{#2}}
\DeclareDocumentCommand{\eqPrice}{O{\price}}{#1^{*}}

\newcommand{\pr}{\mathrm{Pr}_\sigma}
\newtheorem{lemma}{Lemma}
\newtheorem{proposition}{Proposition}

\newtheorem{theorem}{Theorem}

\newtheorem{corollary}{Corollary}

\theoremstyle{definition}
\newtheorem{definition}{Definition}
\newtheorem{example}{Example}

\newlist{secenum}{enumerate}{10}
\setlist[secenum]{label=\thesection.\arabic*,leftmargin=*}

\usepackage{placeins}
\usepackage{longtable}
\usepackage{lscape}
\newcommand{\Params}{n,B,\mu,p,\tau}
\newcommand{\limninf}{\lim_{n\rightarrow\infty}}
\definecolor{ao}{rgb}{0.0, 0.5, 0.0}

\newcommand{\blue}{\textcolor{blue}}

\usepackage{multirow}
\newtheorem{innercustomgeneric}{\customgenericname}
\providecommand{\customgenericname}{}
\newcommand{\newcustomtheorem}[2]{%
	\newenvironment{#1}[1]
	{%
		\renewcommand\customgenericname{#2}%
		\renewcommand\theinnercustomgeneric{##1}%
		\innercustomgeneric
	}
	{\endinnercustomgeneric}
}

\newcustomtheorem{customthm}{Theorem}
\newcustomtheorem{customlemma}{Lemma}
\newcustomtheorem{customprop}{Proposition}
\begin{document}
\title{The One-Shot Crowdfunding Game}
\author{Itai Arieli}
\author{Moran Koren}
\author{Rann Smorodinsky}

\affil{\small Faculty of Industrial Engineering, Technion\textemdash Israel Institute of Technology. \thanks{Smorodinsky gratefully acknowledges the support the joint United States-Israel Binational Science Foundation and National Science Foundation grant 2016734, German-Israel Foundation grant I-1419-118.4/2017, Ministry of Science and Technology grant 19400214, Technion VPR grants, and the Bernard M. Gordon Center for Systems Engineering at the Technion.}}

\renewcommand\Authands{, and }
\renewcommand\footnotemark{}
\maketitle

% note that the abstract must come before \maketitle
\begin{abstract}
	Society uses the following game to decide on the supply of a public good. Each agent can choose whether or not to contribute to the good. Contributions are collected and the good is supplied whenever total contributions exceed a threshold. We study the case where the public good is excludable, agents have a common value and each agent receives a private signal about the common value. This game models a standard crowdfunding setting as it is executed in popular crowdfunding platforms such as Kickstarter and Indiegogo. We study how well crowdfunding performs from the firm's perspective, in terms of market penetration, and how it performs from the perspective of society, in terms of efficiency.
	
\end{abstract}

% note: this command has been disabled to remove the ACM copyright block. Sorry...
\thanks{This work is supported by the National Science Foundation,
	under grant CNS-0435060, grant CCR-0325197 and grant EN-CS-0329609.}

\maketitle

% Renew this after \maketitle if the default list of authors is too long for headers
%\renewcommand{\shortauthors}{W.\ Vickrey et.\ al.}

\section{Introduction}

The evolution of the `sharing economy' has  made it possible for the general public to invest in early-stage innovative and economically risky projects and products. In 2015 the total funds raised via this innovative form of funding, commonly referred to as  \textit{Crowdfunding}, already exceeded 34 Billion Dollars and it is by all means the largest growing avenue for funding new products. Such funding may take the form of a capital investment, peer-to-peer loans, early purchase of goods, typically in a nascent and undeveloped stage and new innovative investment structures such as initial coin offerings.

In traditional funding avenues, the power to decide which projects to support and which products would prevail was often endowed to small committees of experts. In the private sector, banks and private equity funds would endow such decisions to their investment committee while in the public sector such decisions would often be taken by a small group of civil servants and public officials. Crowdfunding essentially revokes the power of such small teams and endows the funding decision to the crowd with the basic premise that the crowd is smarter than any small team of experts. The goal of this paper is to study how well the wisdom of the crowd performs in the context of funding decisions.

Inspired by popular crowdfunding platforms such as `Kickstarter' and `Indiegogo' we introduce a simple game of incomplete information, which we dub the {\it Crowdfunding game}. A firm who want to propose a new product offers the the following game to its potential customer base: The firm posts a price for its product and, in addition, sets a revenue goal.%
\footnote{In reality the firms typically propose more than a single variant of the product alongside a menu of posted prices.}
The product is at a nascent stage and so its true value is yet unknown.
Potential customers may have some private information regarding the value of the product. Based on this information the customers choose whether or not to buy the product at the posted price (hereinafter we refer to this action as a contribution). If the total contributions pledged in the campaign exceed the revenue goal then contributions are collected and the firm supplies the product to the contributors. Otherwise contributions are not collected.

As crowdfunding campaigns are often associated with early stage products, when the demand is unknown, they serve a few  objectives. From the firm's perspective the goal of the crowdfunding campaign is to raise funds in order to develop the product. Equally as important, the campaign serves to raise awareness to the product and so it serves as a means to penetrate the market and provide exposure to a critical mass of early adopters. From society's point of view it serves to aggregate the information from the crowd and so it serves as an institution to tunnel funds to the viable products. Ideally, crowdfunding campaigns will deny funds from low value products and projects while guaranteeing the support to high value products and projects.

Thus, we associate with each game two indices - a (market) Participation Index which is associated with how well do campaigns perform in terms of attracting contributors, and a Correctness Index that is associated with how well do campaigns harness the wisdom of the crowd to support the high quality products while denying funds to the low quality ones.

More technically, the Crowdfunding Game is a game of incomplete information played among a population of  $n$ potential contributors. The common value of the product, $v$, is unknown and players have  some private information about this value. A player must decide whether or not to buy the product at some posted price, $\tau$ (`contribute'). If a player contributes and the total number of contributors exceeds some preset threshold, $B$, then her utility is $v-\tau$.%
\footnote{Alternatively, one can set the threshold in terms of contributions pledged and not in terms of the number of contributors. We chose the latter form as it seems that market traction often plays a more important role than actual revenues.}
Otherwise it is zero. In particular foregoing the contribution opportunity entails a utility of zero.
The two measures of success for a crowdfunding game that we study are:
\begin{itemize}
	\item
	The {\em correctness index} of a game is defined as the probability that the game ends up doing the correct thing. That is, the probability the product be funded when its value is high or is rejected when its value is low. The correctness index measures how well the crowdfunding aggregates the private information from the buyers in order to make sure the firm pursues the product only when it is viable.
	\item
	The {\em market penetration index} is the expected number of buyers provided that the product is supplied, i.e, the threshold is surpassed. This number serves as a proxy for success of the campaign as a means to attract further investments.
\end{itemize}

Our proposed crowdfunding game is a stylized model for how crowdfunding actually takes place in reality. One obvious limitation of the current model is that it comprises  of a simultaneous move game, whereas in reality campaigns are executed over a period of time and agents have the option to wait for others (possibly more informed) agents to make a pledge before they commit. This, however, is not a major drawback of the model as it has been noticed empirically that the majority of contributions made by unaffiliated players (no family and friends) take place just before the campaign's deadline (see \cite{Kuppuswamy2018}).

The crowdfunding game is sufficiently simple and abstract to serve as a model for the formation of institutions. For example, consider the evolution of multi-national institutions (e.g., the UN's International Court of Justice in Hague, the Kyoto Protocol or the Geneva Conventions) where these institutions form only if supported by sufficiently many nations and serve the supporting nations only. In a similar vein, the formation of industry standards can be modeled as a Crowdfunding game.

\subsection{Main findings}\label{sec:main_fin}

Our first result establishes the existence of a symmetric, non-trivial equilibrium in crowdfunding games. It turns out that, for some parameter combinations, such an equilibrium necessarily exists while for others it is guaranteed to exist only when the crowd is large enough. Once this has been established we turn to study the consequences of such symmetric equilibria of {\bf large} crowdfunding games in terms of both aforementioned two success measures:

\begin{itemize}
	\item
	We provide a tight bound on the  correctness index which is strictly less than one. Thus, no matter how the campaign goal is set, full information aggregation cannot be guaranteed. We compare this with the efficiency guarantees of majority voting implied by Condorcet Jury Theorem.
	\item
	We provide a bound on the penetration index and we show that by setting the champaign goal optimally the resulting market penetration is higher than the prior.
\end{itemize}

Our analysis is typically done for three distinct cases:
\begin{itemize}
	\item
	Games in which the price is cheap, and players contribute regardless of their personal signal.
	\item
	Games with moderate prices where the only (symmetric) equilibrium is one in which players with a high signal surely contribute while those with a low signal either decline or take a mixed strategy whereby they contribute at a positive probability, strictly less than one; and
	\item
	Games with expensive prices where, for sufficiently low thresholds, the only (symmetric) equilibrium is one in which players with a low signal opt-out and players with a high signal play a mixed strategy.
\end{itemize}

In addition to the aforementioned theoretical results we present some computational results that pertain to symmetric equilibria in moderate size crowdfunding games. Inspired by related empirical research we focus on games with around $100-1000$ players, which is a realstic estimate for real-world crowdfunding campaigns. By and large the computations corroborate that our theoretical findings in the asymptotic analysis prevail in moderate size games.

\subsection{Related Literature}

The lion's  share of the literature on crowdfunding takes an empirical approach according. In this context, one can divide the relevant literature into two strands. One strand uses crowdfunding data to calibrate parameters of some complex systems to match the date best (e.g., \cite{Lei2017},\cite{Ellman2014}, and \cite{Yang2016}). In contrast with our model, the laws of motion for the underlying models in these papers are not derived from strategic analysis of the players and so they are not the result of any equilibrium analysis.%
\footnote{For example, Yang et al \cite{Yang2016} assume that the the decision of every agent is determined by the historic success rate of previous projects.}

In another strand of the relevant empirical literature, data from online crowdfunding platforms is summarized statistically and some overarching observations are made on such campaigns, often in the context of additional variables such as culture and geography (e.g., Hemer \cite{Hemer2011}).
Three of these observations are worth noting in the context of our work: Yum et al. \cite{Yum2012} argue that firms use the crowdfunding platform as a means for information gathering. Mollick \cite{Mollick2014} observes that most crowdfunding campaigns reach extreme results. Either, the number of contributors to a campaign is small or it is over subscribed. The same author  uses a survey of over $47,000$ contributors to conclude that about $9$ percent of successful campaigns never deliver (\cite{Mollick2015}).

Recent empirical papers (\cite{Kuppuswamy2017}, \cite{Kuppuswamy2018}) study data from $14,704$ ``Kickstarter" campaigns held between 2012 and 2014 and provide new insights into crowdfunding campaigns: (1) The magnitude of contributions is greater in the first and last week of a campaign's time span yielding a ``U-shape" pattern over time. This U-shape is seen both in successful and failed campaigns (by `failure' we mean a campaign for which the contributions fell short of the threshold). (2) In most cases, failed campaigns fail by a large margin while successful ones succeed only by a thread. (3) Once the campaign goal is reached, the rate of contributions decreases significantly.
In \cite{Kuppuswamy2017}, the authors go on and provide a behavioral model that is compatible with the data, but does not assume agents are rational (or common knowledge thereof). For example, in the proposed model agents are over confident of their influence on the campaign outcome, an observations that is often inconsistent with equilibrium analysis and common knowledge of rationality.

\cite{Kuppuswamy2018} examines the U-shape contributions pattern and find that the early backers are not necessarily playing to maximize their value as they primarily belong to the social circle of the entrepreneur (friends and family). On the other hand, most of the activity of unaffiliated backers takes place at the very end of the campaign. Cating this observation onto our model suggests that our crowdfunding game should be thought of as a model of the final stage of the campaign, when value-maximizing agents take action.

Game theoretical models have been used to study a variety of aspects of crowdfunding.
Strausz \cite{Strausz2017}  studies the vulnerability of crowd-funding platforms to entrepreneurial moral hazard. In contrast with our model the firm has the informational advantage and may seek to embezzle part of the funds. The paper offers an efficient mechanism to circumvent this issue. Chemla and Tinn \cite{Chemla2016} compare two common crowdfunding mechanisms - ``All-or-Nothing (AoN)" and ``Keep-it-All (KiA)". In AoN, as in our model, funds are collected only if a pre-determined threshold is reached. In the KiA mechanism this threshold is set to zero. The paper shows that AoN dominates KiA in terms of efficiency and is less vulnerable to moral hazard. Kumar et al \cite{Kumar2017} study the competition between two means for fund raising -  crowdfunding and loans. They go on and show the connection between the cost of capital, the level of price discrimination in the crowdfunding campaign and  the efficiency of the final allocation.

Finally, Alaei, Malekian and Mostagir \cite{Alaei2016}, consider a model of crowdfunding where buyers with private valuations take actions sequentially. Whereas their model is not strategic, and players follow some ad-hoc `natural' strategy, their conclusion supports the main finding in Mollick  \cite{Mollick2014}. Namely, crowdfunding campaigns most often end in one of two extreme outcomes, they either attract a few contributors or are oversubscribed.

%We study the efficiency in which the crowdfunding platform aggregates information from the individual agents. Our setting is one where the proposed product has an unknown common value. We focus on  issues related to the aggregation of information. These issues have been overlooked in the past and are ignored in \cite{Alaei2016}.

Somewhat related to our model is the line of research on the Condorcet model. In the standard model, similar to the crowdfunding game, players have a state dependant common value with some private information and ea player can take one of two actions (`vote'). The Condorcet Jury theorem argues that the majority rule will aggregate information. In other words it will result in the correct decision if voters vote naively (`truthfully') and the population is large enough (this is no more than the law of large numbers).
Austen-Smith and Banks \cite{Austen-Smith1996} challenge this premise by noting that naive voting is not necessarily rational. Mclennan \cite{McLennan1998} provides an alternative framework where Condorcet's asymptotic efficiency results hold in equilibrium. In contrast with the Condorcet Jury theorem and Mclennan's result the crowdfunding game need not aggregate information fully and could lead to an inefficient outcome, even in large populations.

The paper is organized as follows. In Section \ref{sec:model} we present the Crowdfunding game and our first result regarding existence and uniqueness of an equilibrium.
In section \ref{sec:results}  we present our asymptotic results for the Crowdfunding game, related to large markets. In Section \ref{sec:finite} we provide some calculations for the outcome of such games in smaller markets. We conclude in Section \ref{sec:conc} and suggest some future avenues of further investigation.

\section{The Crowdfunding Game}\label{sec:model}

A crowdfunding game is a game of incomplete information played among a population of $n$ potential contributors (or players). An unknown state of nature $\omega \in \Omega=\{H,L\}$ is drawn with prior probabilities $(\mu,1-\mu)$, respectively.
In state $\omega$ the common value of the product is $v_\omega.$ Conditional on the realized state $\omega$, a private signal $s_i\in S_i=\{h,l\}$ is drawn independently for every player $i$. We assume $p= Pr(s_i=h|\omega=H)=Pr(s_i=l|\omega=L)> 0.5$. Each player $i$ has a binary action set, $A_i= \{0,1\}$,
with $a_i=1$ representing a decision to contribute. A contribution can be seen, in fact, as a commitment to buy the product at some pre-set  price, $\tau$, if it is eventually supplied. The action $a_i=0$ represents a decision to opt-out and not to contribute.  The utility of every player $i\in N$ is defined as follows
\begin{equation}\label{eq:consumer_util}
u_i(a_i,a_{-i},\omega)=\begin{cases}
v_H-\tau&\mbox{ if }a_i=1\mbox{ and }\sum_{j=1}^n a_j \geq B\mbox{ and }\omega=H\\
v_L-\tau&\mbox{ if }a_i=1\mbox{ and }\sum_{j=1}^n a_j \geq B\mbox{ and }\omega=L\\
0&\mbox{ otherwise}
\end{cases}.
\end{equation}
In words, whenever player $i$ chooses to opt-out, she receives a utility of zero. If she chooses to contribute, then her utility is determined by the total number of contributors. If less than $B$ players contributed then the product is not supplied and her utility is zero. If the number of contributions exceeds $B$ then her utility is determined by the state of nature and equals $v_H-\tau$ in state  $H$ and $v_L-\tau$ in state $L$. Hereinafter we assume, without loss of generality, that $0=v_L< \tau < v_H=1$ and denote the corresponding game by $\Gamma(n,B,\mu,p,\tau)$.% \footnote{When the context allows, we omit the notation of game parameters and denote the game by $\Gamma.$}

A strategy for player $i$ is a mapping $\sigma_i:S_i\rightarrow \Delta A_i$. For simplicity we identify $\sigma_i(s)$ with the probability that player $i$ assigns to the action $1$ (`contribute'),  conditional on receiving signal $s$. A strategy profile is called \textit{symmetric} if $\sigma_i=\sigma_j$ for all players $i,j\in N$.

The distribution over the states of nature and the corresponding vector of signals, coupled with a strategy profile, $\sigma$, induce a probability distribution over the players' actions profile. A strategy profile $\sigma$ forms a \emph{Bayes-Nash equilibrium} if
$$E_{\sigma}(u_i(\sigma_i(s_i),\sigma_{-i}(s_{-i}))) \ge  E_{\sigma}(u_i(a_i,\sigma_{-i}(s_{-i})))
\ \  \forall i,\forall s_i \in S_i, \forall a_i\in A_i,$$
where the expectation is taken w.r.t to the aforementioned probability distribution.

One obvious equilibrium in the crowdfunding game, whenever $B>1$, is one where all players choose to opt-out, in which case the revenue goal is never met and the product is never supplied. To avoid such equilibria we restrict attention to equilibria for which there is a positive probability that the good be supplied:
\begin{definition}
	A strategy profile (in particular an equilibrium strategy profile) $\sigma=(\sigma_1,\dots,\sigma_i,\dots,\sigma_{n})$ is called \textit{non-trivial} if, $$Pr_{\sigma}(\sum_{i} a_i \geq B)>0.$$
\end{definition}

Our first result is related to the existence and uniqueness of non-trivial,symmetric equilibria. We show that in any crowdfunding game, there can be at most one such equilibrium. Furthermore, when the population is large enough, such an equilibrium is guaranteed to exist.
%This uniqueness is crucial in our asymptotic analysis which follows.

\begin{theorem}\label{thm:gen_atmost_one_eq}
	(1) No crowdfunding game has more than one symmetric non-trivial Bayes-Nash equilibrium. (2) Consider the sequence $\{B_n\}_{n=1}^{\infty}$ where $\lim_{n\rightarrow\infty}\frac{B_n}{n}=q$ for some $q\in(0,1].$   For any 4-tuple of parameters $(q,\mu,p,\tau)$ there exists some $N$ such that for any $n>N,$  the crowdfunding game $\Gamma(n,B_n,\mu,p,\tau)$ has a unique symmetric non-trivial Bayes-Nash equilibrium.
\end{theorem}
%In the sequel we present three theorems (Theorems \ref{thm:eq_asym_game_towards},\ref{thm:unique_eq}, \ref{lem:complement_11}), corresponding to  three different price levels, which subsume the above theorem.
%
The proof of Theorem \ref{thm:gen_atmost_one_eq} is relegated to Appendix \ref{sec:proofs}.

\subsection{Performance Measures}
In the introduction we discuss the various objectives of crowdfunding campaigns.
% These two objectives corespind according to two , firms turn to utilize the crowdfunding scheme for two distinct objectives. The first objective is to estimate the market demand for their product, and the second is to signal quality for future investors. The difference between these objectives is that in the former, it is in the firm's best interest to shelve the product if it is bad, while in the latter, the firm is less concerned with the true value of the product and more on the success of the campaign as the assumed risks are to be divided to future investors. To gain a complete view over the efficiency of the crowdfunding game we suggest two measures.
The following two indices correspond to two of these objectives. The first is the \emph{correctness index} of a game which pertains to  how well the game tunnels funds. The second is the \emph{participation index} of a game which refers to how well does the campaign attract contributions. Formally, let  $\sigma^*$ denote the unique symmetric non-trivial equilibrium of the game $\Gamma(n,B,\mu,p,\tau)$. Then:

{\bf The Correctness Index} is the following expectation:
\begin{equation}
\theta(n,B,\mu,p,\tau)=\mu Pr_{\sigma^*}(c^H_n \geq B)+(1-\mu) Pr_{\sigma^*}(c^L_n < B)
\end{equation}
where $c^\omega_n=\sum_{i=1}^n a_i$ is the expected number of contributors, conditional on the realized state $\omega\in\{L,H\}$.  The first summand captures the probability of a correct outcome whenever the state of the world if $H$ and ideally the product should be supplied and the second summand captures the opposite situation.
 If no such equilibrium exists then we set $\theta(n,B,\mu,p,\tau)=0.$

The Correctness index of a {\it large} crowdfunding game, associated with the parameters $(\mu,p,\tau)$ is
$$\theta(\mu,p,\tau) = \lim_{n\rightarrow\infty} \max_{B\in\{1\dots n\}} \theta(n,B,\mu,p,\tau).$$
\\

{\bf The Participation Index} is the following expectation:
\begin{equation}
R(n,B,\mu,p,\tau)=E_{\sigma^*}\big[\frac{c_n}{n} \chi(c_n \geq B) \big] =
Pr_{\sigma^*}(c_n \geq B)E_{\sigma^*}\big[\frac{c_n}{n}|c_n \geq B \big].
\end{equation}
Where $c_n$ counts the number of contributors and $\chi(A)$ is the indicator function of the event $A$. If no such equilibrium exists then we set $R(n,B,\mu,p,\tau)=0$. In words, the Participation Index is the expected number of contributions {\bf collected} (conditional on the campaign target being met).

The Participation index of a large crowdfunding game, associated with the parameters $(\mu,p,\tau)$ is
$$R(\mu,p,\tau) = \lim_{n\rightarrow\infty} \max_{B\in\{1\dots n\}} R(n,B,\mu,p,\tau).$$
\\

\subsection{The role of $B$}

The threshold $B$ that is prevalent in many crowdfunding campaigns (often presented in terms of revenues and not in terms of contributors) plays a dual role. From the society's perspective, it introduces a barrier to entry, guaranteeing funds only to those project with sufficient public support. The underlying implicit assumption is that public support will only be provided whenever the collective wisdom assigns high probability to the state $H$. In addition, from the firm's point of view it serves to entice participants. The fact that a contribution is collected only when the overall support is high enough offers an inherent `social' insurance. That is, when a certain participant is contemplating whether to contribute he does not base his decision only on his private information but also on the likelihood the product is good product, conditional that the threshold $B$ is reached. Consequently, players that are initially doubtful (those with a low signal) will also contemplate a contribution. However, on the other hand, a high threshold may imply lower participation even if more players contemplate a contribution. This is because our notion of participation refers to the number of contributions that are actually collected.
\\
Let us now see how these arguments play out in an example.
\begin{example}\label{example:1}
Consider the following symmetric 3-player crowdfunding game:  $\Gamma(n=3,B, \mu=0.5, p=0.75,\tau=0.5)$:
\begin{itemize}
	\item We first assume a low threshold, $B=1$. Such a low threshold offers no social insurance and each contribution is necessarily collected. In this case the expected utility of a player from contributing is
	$Pr(\omega=H|s_i)(1-\tau)-Pr(\omega=L|s_i)\tau$ which is equal $0.25 >0$
	whenever a player receives a high signal ($s_i=h$) and $-0.25<0$ whenever he receives a low signal. The participation index is therefore the expected proportion of high signals, which is equal $0.5$ and the correctness index is
	$0.5(1-0.25^3)+ 0.5 (0.75)^3= 0.703$.
	\item
	In contrast, consider the high threshold, $B=3$.
	As above, players receiving the high signal will surely contribute. However, having only the high signal players contribute is no longer an equilibrium because of the social insurance effect. That is, if only high signal players contribute, then a low signal player is better-off contributing as in this case he assigns a probability of $0.75$ to $\omega=H$ \emph{conditional on reaching the threshold $B=3$}. However, if all low type player choose to contribute, then the social insurance is no longer valid. Thus, in equilibrium, they use it with caution, or more formally play a mixed strategy. The actual probability of contribution for the low signal players in equilibrium turns out to be $\lambda=0.302$. Now the probability of a successful campaign conditional on the state being $H$ is $0.563$ whereas the probability of a failed campaign conditional on state $L$  is $0.892$. From this we can compute that the correctness index is equal $0.727$, higher than the case $B=1$. On the other hand the participation index is now $0.335,$ lower than the case $B=1.$
\end{itemize}
\end{example}

\section{Asymptotic Results}\label{sec:results}

We present results for three distinct scenarios, depending on the product pricing. We distinguish between three price levels: cheap, moderate and expensive, formulated as follows.
Let $p_l=Pr(\omega=H|s_i=l)= \frac{(1-p)(1-\mu)}{p\mu+(1-p)(1-\mu)}$ and $p_h=Pr(\omega=H|s_i=h) = \frac{p\mu}{p\mu+(1-p)(1-\mu)}$ be the two possible posterior expectations over the value of good, depending on the signal received. Obviously $p_h>p_l$. Recall the values of the product at the two states, $v_H=1$ and $v_L=0$, which in turn implies that the posterior forms the maximal price an agent would pay for the good in a simple take-it-or-leave setting.%
%\footnote{Recall that when a player contributes, his action is binding only if there are enough contributions collected. We distinguish here between this and the ``buy" action, where we implicitly say that the agent knows his contribution will surely be collected.}
The three cases we study are:
\begin{itemize}
	\item The campaign offers a {\em cheap} price whenever $\tau \le p_l <p_h$. It should not be surprising that when the campaign offers a cheap price both types of agents necessarily contribute. Consequently, the crowd does not convey its wisdom.
	\item The campaign offers a {\em moderate } price whenever $ p_l < \tau < p_h $. Whereas players a low signal would not buy the good they nevertheless participate in the campaign (recall Example \ref{example:1}).
	\item The campaign offers an {\em expensive} price whenever $ p_l < p_h \le \tau$. Whereas both types would decline to but the good at the price $\tau$ participation does take place due to the inherent social insurance.%
\end{itemize}

\subsection{Cheap Prices}\label{sec:asym_game}

A Crowdfunding game is  \textit{cheaply priced} whenever $\tau<Pr_{\mu}(\omega=H|s_i=l).$ In such games the outcome is trivial as the unique symmetric equilibrium (which is necessarily non trivial) is for all players (low and high) to participate:

\begin{theorem}\label{thm:eq_asym_game_towards}
	In any crowdfunding game  with a cheap price there exists a unique symmetric Bayesian Nash equilibrium where all players contribute. This equilibrium is non-trivial.
\end{theorem}

The proof of Theorem \ref{thm:eq_asym_game_towards} is relegated to Appendix \ref{sec:proofs} however the intuition behind it is quite straightforward. Whenever the price is cheap both types of players are happy to buy it even if their contributions will surely be collected, and do not require the social insurance for that.

Given the simplicity of the equilibrium strategies we can easily derive the value of the two indices:
\begin{theorem}
	For any crowdfunding game, $\Gamma(n,B,\mu,p,\tau)$, with a cheap price:
	\begin{itemize}
		\item $\theta(n,B,\mu,p,\tau)=\mu$; and
		\item $R(n,B,\mu,p,\tau)=1.$
	\end{itemize}
\end{theorem}

\begin{proof}
	The proof follows immediate from Theorem \ref{thm:eq_asym_game_towards} and the corresponding definitions of the two indices
\end{proof}

\subsection{Moderate Prices}\label{sec:sym_game}

When a campaign price is moderate, high type players find the price attractive while low type players do not:
$$Pr(\omega=H|s_i=l)=  < \tau < Pr(\omega=H|s_i=h) .$$

We begin by establishing the existence and uniqueness of a symmetric non-trivial equilibrium:

\begin{theorem}\label{thm:unique_eq}
	For any crowdfunding game, $\Gamma(n,B,\mu,p,\tau)$, with a moderate price, there exists a \emph{unique} symmetric non-trivial Bayesian Nash equilibrium $\sigma^*=(\sigma^*_1,\ldots,\sigma^*_n).$ Moreover, $\sigma^*_i$ has the following form,
	\begin{equation}\label{eq_equilibrium}
	\sigma^*_i(s_i)=\begin{cases}
	1&\mbox{ if }s_i=h\\
	\lambda=\lambda(n,B,\mu,p,\tau)\in[0,1) &\mbox{ if }s_i=l.
	\end{cases}.
	\end{equation}
\end{theorem}

We relegate the proof to the Appendix \ref{sec:proofs} but provide some intuition. We refer to a player who receives the signal $h$ as a `high' player and to a player who receives the signal $l$ as a `low' player. The high player is perfectly happy with the price and would contribute even without the social insurance embedded in the threshold $B$. What about `low' players?  Assume only high players contribute and none of the low players do. Then each low player has an incentive to leverage the social insurance by deviating and contributing. If, on the other hand, all low players as well as all high players contribute then there is no social insurance and each low player can profitably deviate by opting out. By properly mixing between the two actions each low player can be made indifferent and hence best-replies by mixing. This establishes the equilibrium.

In the following lemma we characterize the limit equilibrium strategy of the low player as the size of the population increases. We restrict the analysis to sequences of games where the the limit, per-capita, threshold exists ($\exists \ \lim_{n\rightarrow\infty} \frac{B_n}{n}=q$ for some $q\in[0,1]$).

%Let $\{\Gamma(n,B_n,\mu,p,\tau)\}_{n=2}^{\infty}$ be a sequence of Crowdfunding games with a moderate price and let $\{\sigma^*_n\}_{n=2}^{\infty}$ be the corresponding sequence of symmetric non-trivial equilibria (guaranteed to exists by Theorem \ref{thm:unique_eq}). Assume that the sequence of threshold $\{B_n\}_{n=1}^{\infty}$ satisfies $\lim_{n\rightarrow\infty}\frac{B_n}{n}=q\in(0,1)$ ().

\begin{lemma}\label{lem:asympt_probs}
	Let $\{\Gamma(n,B_n,\mu,p,\tau)\}_n$ be a sequence of moderately priced crowdfunding games such that $\lim_{n\rightarrow\infty} \frac{B_n}{n}=q$ for some $q\in[0,1]$. Then the limit equilibrium strategy is:
	\begin{equation}
	\lim_{n\rightarrow\infty} \sigma^*_n(l)=\begin{cases}
	0&\mbox{ if } q\leq 1-p\\
	\frac{q-(1-p)}{p}&\mbox{ otherwise}
	\end{cases}
	\end{equation}
\end{lemma}

With this computation at hand we can now turn to study the Correctness index for large markets:
%For symmetric games we often emit the reference to the probabilities and write $\theta(n,B)$ instead of $\theta_{0.5}(n,B).$

Our second result characterizes the asymptotic correctness of the moderate pricing Crowdfunding game.
\begin{theorem}\label{thm:asym_correctness}
	For any large crowdfunding game with prior $\mu$, signal quality $p$ and a moderate price $\tau$ the probability of making the correct choice is given by:
	\begin{equation}\label{eq:max}
	\theta(\mu,p,\tau) =1-\frac{1-p}{p}\frac{1-\tau}{\tau}\mu.
	\end{equation}
\end{theorem}

In fact, a careful reading of the proof suggests that the following slightly stronger result holds. Fix the prior $\mu$, signal quality $p$ and a moderate price $\tau$. If for every $n$ the threshold $B_n$ satisfies ???? (MORAN LEASE FILL IN) then  $\lim_n \theta(n,B_n,\mu,p,\tau) = 1-\frac{1-p}{p}\frac{1-\tau}{\tau}\mu$. To see why this is a bit stronger recall that the definition of $\theta(\mu,p,\tau)$ pertains to the threshold $B$ that maximizes the correctness indices along the sequence and not to arbitrary thresholds.

An immediate conclusion is that large crowdfunding campaigns, in the format we study, necessarily exhibit market failure when prices are moderate.%
\footnote{Compare this observation with Condorcet's jury theorem which argues that in a majority vote, large societies necessarily choose the correct alternative.}
This failure probability is given by $\frac{1-p}{p}\frac{1-\tau}{\tau}\mu$ and the following comparative statics follow immediately:

\begin{corollary}.
For any large crowdfunding game with prior $\mu$, signal quality $p$ and a moderate price $\tau$ the market failure probability decreases as one of the following occurs: (1) the signal accuracy of the signal increases, (2) the price increases; and (3) the prior probability (for the value being high) decreases.%
\footnote{Note the the constraint that prices are moderate rules out the extreme case $\mu=1$ in which it would have been surprising to learn of the possibility of market failure.}
\end{corollary}.

Below we provide an outline of the proof of Theorem \ref{thm:asym_correctness} while relegating the full proof to  Appendix \ref{sec:proofs}. The proof leverages the intuition hinges that any player, conditional on the actual state of nature, is (almost) non-pivotal (similar to \cite{AlNajjar2000} and \cite{Levine1995}). In other words, whenever the population is large enough, each individual player deems her own action to have impact on the probability of supply, conditional on knowing the state of nature. Thus, in each state $\omega$ the probability of supply, given her contribution and the state $\omega$ is approximately equal $Pr(c^{\omega}_n\geq B)$.%
%We are now ready to outline the proof of Theorem \ref{thm:asym_correctness}.

\noindent
{\bf Proof Outline of Theorem \ref{thm:asym_correctness}:}
Consider the sequence of games $\{\Gamma(n,\frac{n}{2},\mu,p,\tau)\}_{n=2}^{\infty}.$
By Lemma \ref{lem:asympt_probs}, the corresponding sequence of equilibrium strategies for low players converges to
\begin{equation}\label{eq:assym_lambda}
\limninf\sigma^*_n(l) = \frac{2p-1}{2p}.
\end{equation}
%By equation \eqref{eq:assym_lambda}, for  sufficiently large $n$, $\lambda(n,\frac{n}{2},\mu,p,\tau)>0$, and thus by the indifference condition for player of type $l$,
%$$\frac{(1-p)\mu x(n,B_n)(1-\tau)-p(1-\mu)y(n,B_n)\tau}{(1-p)\mu+p(1-\mu)}=0.$$

Let $\alpha^{\tilde\omega}=Pr_{\sigma}(a_i=1|\omega=\tilde{\omega})$ be the probability that an arbitrary player contribute in the state ${\tilde\omega}$. Using equation \eqref{eq:assym_lambda} and the fact that high players necessarily contribute we get that $\alpha^H$ converges to $\frac{3p-1}{2p}>\frac{1}{2}.$ This implies that the probability for a successful campaign, conditional on the state $H$ approaches one.

Using similar computations and relying on the indifference of the low players, we can show that whenever the state is $L$ the probability of success approaches
$ \frac{1-p}{p}\frac{\mu}{1-\mu}\frac{1-\tau}{\tau}.$  Combining these two computations yields a lower bound:
$$\theta(\mu,p,\tau) \ge 1-\frac{1-p}{p}\frac{1-\tau}{\tau}\mu.$$

To show the opposite inequality consider an arbitrary sequence $\{B_n\}$ and assume that the following three sequences converge:
$\{\theta(n,B,\mu,p,\tau)\}, \{Pr(c^H_n\geq B \}_n$ and $\{Pr(c^L_n\geq B)\}_n$ (otherwise, consider a sub-sequence). Let us denote the corresponding limits by $\theta^*, x^*$ and $y^*$. By the definition of the correctness index and by Lemma \ref{lem:asympt_probs} we get,
\begin{equation}\label{eq:correct_ll1}
\theta^* = \mu x^* + (1-\mu)(1-y^*).
\end{equation}

Recall that a `low' player mixes and so is indifferent between the two actions. Taking the limit of the indifference equation of the `low' players yields:
\begin{equation}\label{eq:correct_ll2}
(1-p)\mu x^*(1-\tau)-p(1-\mu)y^*\tau=0.
\end{equation}
By equations \eqref{eq:correct_ll1} and \eqref{eq:correct_ll2},  the asymptotic correctness value is bounded above by the solution for the following linear program:
\begin{equation}
\begin{aligned}
& {\text{max}}
& & \mu x^*+(1-\mu)(1-y^*) \\
& \text{s.t.} & & 1\geq x^*,y^*\geq 0 \\
& & &  (1-p)\mu x^*(1-\tau)-(1-\mu)py^*\tau=0, \\
\end{aligned}
\end{equation}
which is $1-\frac{1-p}{p}\frac{1-\tau}{\tau}\mu$.
\qed

We now to compute the Participation index for large markets when prices are moderate:

\begin{theorem}\label{thm:revenue} For any large crowdfunding game with prior $\mu$, signal quality $p$ and a moderate price $\tau$, the participation index is given by:
	\begin{equation}\label{eq:revenue}
	R(\mu,p,\tau) =\mu(1+\frac{1-p}{p}\frac{1-\tau}{\tau})
	\end{equation}
\end{theorem}

Note that participation is always greater the $\mu.$ Furthermore it increases as the price decreases and as the prior (for the good state) increases. Perhaps less intuitive is the conclusion that penetration decreases as the signal, $p$, becomes more accurate. A possible explanation is that with less accuracy the `low' players put more emphasis on the aforementioned social insurance. This is manifested in equation \eqref{eq:assym_lambda} which shows that the contribution probability of such players increases in  $p$.

\noindent
{\bf Proof Outline of Theorem \ref{thm:revenue}:} Let $\{\Gamma(n,qn,\mu,p,\tau)\}_{n=1}^{\infty}$ be a sequence of crowdfunding games. We discuss the two different cases, $q >1-p$ and $q\le 1-p$, separately.

\noindent
{\bf Case 1, $q >1-p$:}  By Lemma \ref{lem:asympt_probs} \begin{equation}\label{eq:indif_limit}
\limninf\sigma^*_n(l)=\frac{q-(1-p)}{p}>0.
\end{equation}
Let $\alpha^{\tilde\omega}_n = Pr_{\sigma^*_n}(a_i=1|\omega=\tilde{\omega})$ be the probability that an arbitrary player contributes in the state $\tilde{\omega}.$ Using equation \eqref{eq:indif_limit} and the fact that high players necessarily contribute we get that $\alpha^H_n$ converges to $p+(1-p)\frac{q-(1-p)}{p}>q.$ This implies that the probability for a successful campaign, conditional on the state $H$ approaches one.
Using similar computations and relying on the indifference of the low players and the observation that in large games, players are (almost) non-pivotal, we can show that whenever the state is $L,$ the probability of success approaches $\frac{1-p}{p}\frac{\mu}{1-\mu}\frac{1-\tau}{\tau}.$
Combining these two computations yields the unconditional probability of a successful campaign:
\begin{equation}\label{eq:mod_price_limit_sucess_prob}
\begin{split}
&\lim_{n\rightarrow \infty} Pr(c_n\ge B_n)=\lim_{n\rightarrow \infty} Pr(\frac{c_n}{n}\ge q)=\\
&\mu+(1-\mu)\frac{1-p}{p}\frac{\mu}{1-\mu}\frac{1-\tau}{\tau}
=\mu(1+\frac{1-p}{p}\frac{1-\tau}{\tau}).
\end{split}
\end{equation}

The Participation index is the expected number of contributions conditional on the campaign's success and therefore
$$
R(\mu,p,\tau)\ge \lim_{n\rightarrow \infty} R(n,B_n,\mu,p,\tau)\ge \lim_{n\rightarrow \infty} Pr(c_n\ge B_n) =  \mu(1+\frac{1-p}{p}\frac{1-\tau}{\tau}).
$$

In addition, for any $q$ the expected number of contribution conditional on a successful campaign is bounded below:
$$
\lim_{n\rightarrow \infty} R(n,B_n,\mu,p,\tau)\ge q(\mu(1+\frac{1-p}{p}\frac{1-\tau}{\tau})).
$$

Maximizing over $q>1-p$ yields $R(\mu,p,\tau)\ge \mu(1+\frac{1-p}{p}\frac{1-\tau}{\tau})$, as desired.

\noindent
{\bf Case 2, $q \le 1-p$:}
By Lemma  \ref{lem:asympt_probs} `low' players opt-out and only high players contribute. Thus, the expected number of contributions conditional on success equals the expected number of `high' players which yields an upper bound,
$$
\lim_{n\rightarrow \infty}R(n,qn,\mu,\tau)\le p\mu+(1-p)(1-\mu).
$$
As prices are moderate, this implies $p\mu+(1-p)(1-\mu)\le\mu(1+\frac{1-p}{p}\frac{1-\tau}{\tau})$, in the second case.

Combining the two cases yields the desired result.
\qed

\subsection{Expensive prices}

The price in the crowdfunding game is expensive whenever it is high enough such that none of the players would buy it without any additional insurance. Formally, $\tau > Pr_{\mu}(\omega=H|s_i=h).$

In contrast with the two previous cases, a symmetric non-trivial equilibrium need not exist when the population is small.
% In fact, Lemma \ref{lem:no_sigma_small_priors} in Appendix \ref{sec:proofs_small_mu}, argues that keeping the parameters $(n,B,\mu,p)$ fixed, there is a sufficiently high price $\hat\tau$ such the corresponding game, $\Gamma(n,B,\mu,p,\tau)$ has no such equilibria for $\tau>\hat\tau$.
However, by Theorem \ref{thm:gen_atmost_one_eq}, when we consider large games, existence of exactly one non-trivial, symmetric, Bayes-Nash equilibrium is guaranteed.

\begin{theorem}\label{lem:complement_11}
	Let $\{\Gamma(n,B_n,\mu,p,\tau)\}_n$ be a sequence of expensively priced crowdfunding games such that $\lim_{n\rightarrow\infty} \frac{B_n}{n}=q$ for some $q\in[0,1]$. Then the limit equilibrium strategy is:
	\begin{equation}
	\lim_{n\rightarrow\infty} \sigma^*_n(l)=\begin{cases}
	0&\mbox{ if } q\leq 1-p\\
	\frac{q-(1-p)}{p}&\mbox{ otherwise}
	\end{cases}\mbox{ and }
	\lim_{n\rightarrow\infty} \sigma^*_n(h)= 	
	\begin{cases}
	\frac{q}{1-p}&\mbox{ if } q\leq 1-p\\
	1&\mbox{ otherwise}
	\end{cases}
	\end{equation}
\end{theorem}

%Note, in particular that this holds when prices are expensive.
The proof of Theorem \ref{lem:complement_11} is relegated to Appendix \ref{sec:expensive_price_proofs}.

By Theorem \ref{lem:complement_11}, the equilibrium strategy depends on how high the threshold is. That is, even when the price is expensive, the equilibrium may take a similar form as that of the moderate price case whereby high players necessarily contribute while low players mix. However, for certain threshold levels we observe a different form of equilibrium, whereby low players opt-out while high players rely on the social insurance and mix.

The value of the two indices is given in the two last theorems. The main ideas underlying these proofs are similar to the analysis of the moderate price case.

\begin{theorem}\label{thm:correct_asym_game_e_price}
	\begin{equation}
	\theta(\mu,p,\tau)=
	1-\frac{1-p}{p}\frac{1-\tau}{\tau}\mu.
	\end{equation}
\end{theorem}

\begin{theorem}\label{thm:part_asym_game1}\hfill
	\begin{itemize}
		\item  If $\mu<\frac{1}{3}$ and $p\le\sqrt{3}-1,$ or if $\mu<\frac{1}{3}, p>\sqrt{3}-1$ and $\tau>\frac{2\mu }{(1-\mu)p+2(1-p)\mu)}$ then	 	
		\begin{equation*}
		R(\mu,p,\tau)=\mu p+(1-\mu)\frac{1-p}{2}=\frac{(3\mu-1)p+(1-\mu)}{2}.
		\end{equation*}
		\item Otherwise,	
		\begin{equation*}
		R(\mu,p,\tau)=
		\mu(1+\frac{1-p}{p}\frac{1-\tau}{\tau})
		\end{equation*}	
	\end{itemize}
\end{theorem}

Note that participation index is more than the prior, $\mu$. I addition, note that whenever the equilibrium takes the form where only high players participate the participation decreases as signals become more accurate.

MORAN - ANY MORE INTERESTING COMPARATIVE STATICS? ????????????

\section{Crowdfunding in small populations}\label{sec:finite}

Our theoretical results pertain to the asymptotic case and so are relevant to large markets.  Empirical data suggests that crowdfunding campaigns eventually attract around $100$ contributors   (See \cite{Kuppuswamy2017} and \cite{Mollick2014} ).
For example, in \cite{Kuppuswamy2017}, they find that the average number of contributors to Kickstarter campaigns that took place during March and April of  2012 was $100.32$.  These empirical results  suggest that an initial market size for such campaigns is of the of the order of magnitude of $100-1000$.

In order to validate our theoretical results we compute the equilibrium strategies ( $\sigma^*$) and the value of the to indices for the relevant market size ($n=100$ and $n=1000$). We do so for a variety of parameter values (threshold, signal accuracy and price). The result for a market size of $n=100$ are depicted in Table \ref{tab:n-100} while those for $n=1000$ are depicted in Table \ref{tab:n-1000}. Finally,  Table \ref{tab:n-infty} details the asymptotic results for the corresponding parameter values (recall that the asymptotic results are stated in terms of the optimal threshold).

Our theoretical results for case where prices are cheap hold for any market size and so the calculations we report on below are only for the case of moderate prices ($(p,\tau)$ equal $(0.55,0.5)$,  $(0.75,0.5)$ , $(0.75,0.7)$ ) and high prices ( $(p,\tau)  =(0.55,0.7)$ ).

The tables of results below are partial and, in particular focus on the symmetric prior. The interested reader is referred to Appendix \ref{sec:finite_tab} for the calculations in a wider variety of parameters, including asymmetric priors.

\begin{table}[]
	\centering
	\caption{$\Gamma(n=100,B,\mu=0.5,p,\tau).$ }
	\label{tab:n-100}
	%	\resizebox{\textwidth}{!}{%
	\begin{tabular}{|c|c|c c c c|c c c c|}
		\hline
		$p$ & $B$    & \multicolumn{4}{|c|}{$\tau=0.5$}        & \multicolumn{4}{|c|}{$\tau=0.7$}    \\
		\hline
		&   & $\psi$ & $\lambda$ & $\theta$ & $R$  & $\psi$ & $\lambda$ & $\theta$ & $R$      \\
		\hline
		\multirow{3}{*}{0.55} & 9  & 1   & 0      & 0.5   & 0.5    & 0.126 & 0      & 0.570  & 0.019  \\
		& 44 & 1   & 0.044  & 0.606 & 0.478    & 0.957 & 0      & 0.751  & 0.366       \\
		& 98 & 1   & 0.954  & 0.561 & 0.594      & 1     & 0.866  & 0.518  & 0.037     \\
		\hline
		\multirow{3}{*}{0.75}  & 9  & 1   & 0      & 0.5 & 0.5     & 1     & 0      & 0.5  & 0.5            \\
		& 44 & 1   & 0.211  & 0.854 & 0.469       & 1     & 0.170  & 0.941  & 0.571 \\
		& 98 & 1   & 0.951  & 0.795 & 0.575      & 1     & .931  &0.823 & 0.422\\
		\hline
	\end{tabular}%
	%	}
\end{table}

\begin{table}[]
	\centering
	\caption{$\Gamma(n=1000,B,\mu=0.5,p,\tau).$ }
	\label{tab:n-1000}
	%	\resizebox{\textwidth}{!}{%
	\begin{tabular}{|c|c|c c c c |c c c c| }
		\hline
		$p$ & $B$    & \multicolumn{4}{|c|}{$\tau=0.5$}        & \multicolumn{4}{|c|}{$\tau=0.7$}    \\
		\hline
		&   & $\psi$ & $\lambda$ & $\theta$ & $R$  & $\psi$ & $\lambda$ & $\theta$ & $R$      \\
		\hline
		\multirow{3}{*}{0.55} & 90  & 1   & 0      & 0.5   & 0.5    & 0.198 & 0      & 0.752  & 0.07          \\
		& 440 & 1   & 0.006  & 0.596 & 0.462    & 0.977 & 0      & 0.745  & 0.384   \\
		& 980 & 1   & 0.968  & 0.586 & 0.845 & 1      & 0.954      & .668& 0.341\\
		\hline
		\multirow{3}{*}{0.75}
		& 90  & 1   & 0      & 0.5 & 0.5       & 1     & 0      & 0.5  & 0.5          \\
		& 440 & 1   & 0.243  & 0.839 & 0.477      & 1     & 0.23  & 0.932  & 0.434   \\
		& 980 & 1   & 0.970  & 0.834 & 0.659  & 1     & .966  &0.929 & 0.565   \\
		\hline
	\end{tabular}%
	%	}
\end{table}

\begin{table}[]
	\centering
	\caption{Large markets: $\theta(\mu=0.5,p,\tau)$ and $R(\mu=0.5,p,\tau)$ }
	\label{tab:n-infty}
		\resizebox{\textwidth}{!}{%
	\begin{tabular}{|c|c|c|c|c|c|c|c|c|c|}
		\hline
		$p$                   & $q$  & \multicolumn{4}{c|}{$\tau=0.5$}                                        & \multicolumn{4}{c|}{$\tau=0.7$}                                        \\ \hline
		&      & $\psi$ & $\lambda$ & $max\theta(\mu,p,\tau)$ & $max R(\mu,p,\tau)$   & $\psi$ & $\lambda$ & $max\theta(\mu,p,\tau)$ & $max R(\mu,p,\tau)$    \\ \hline
		\multirow{3}{*}{0.55} & 0.09 & 1      & 0         & \multirow{3}{*}{0.590}   & \multirow{3}{*}{0.909} & 0.2    & 0         & \multirow{3}{*}{0.825}   & \multirow{3}{*}{0.675} \\ \cline{2-4} \cline{7-8}
		& 0.44 & 1      & 0         &                          &                        & 0.978  & 0         &                          &                        \\ \cline{2-4} \cline{7-8}
		& 0.98 & 1      & 0.964     &                          &                        & 1      & 0.964     &                          &                        \\ \hline
		\multirow{3}{*}{0.75} & 0.09 & 1      & 0         & \multirow{3}{*}{0.833}   & \multirow{3}{*}{0.667} & 1      & 0         & \multirow{3}{*}{0.929}   & \multirow{3}{*}{0.571} \\ \cline{2-4} \cline{7-8}
		& 0.44 & 1      & 0.253     &                          &                        & 0.978  & 0.253     &                          &                        \\ \cline{2-4} \cline{7-8}
		& 0.98 & 1      & 0.973     &                          &                        & 1      & 0.973     &                          &                        \\ \hline
	\end{tabular}
		}
\end{table}

There are several observations to be made from these tables (and from the additional calculations reported in  Appendix \ref{sec:finite_tab}):
\\
\begin{itemize}
	\item	The most interesting observation is that our asymptotic analysis provides a good approximation for Crowdfunding games with a realistic market size. The computed strategies converge quite fast and the corresponding bounds on the two indices are already quite relevant for these values. This holds both when prices high and more so when prices are moderate. This observation is robust with respect to the value of the signal accuracy and the price.
	\item
	Additionally, we observe that when the signal is weak ($p=0.55$) and the threshold is low ($B\approx\frac{n}{3}$),  low type players always opt-out. Nevertheless, the number of high type players is sufficient to induce production even if $\omega=L$ which causes a rapid deterioration of  the Correcteness index.  As $B$ increases the risk facing low type players decreases and therefore we can see that throughout the table, a higher threshold  $B$ leads to a higher probability that a low player will contribute (a higher $\lambda$)
	\item
	As predicted the correctness index is an increasing with price and decreasing with $\mu.$  Similarly the participation index is an increasing function of the prior $\mu$ and  a decreasing function of prices. This can be verified in Table \ref{tab:good_appx} bellow.
	\item
	Surprisingly, for a variety of parameter combinations the  theoretical predictions are quite accurate even for a very small populations, $n\in\{5,10\}$ (see the table in Appendix \ref{sec:finite_tab} ).
\end{itemize}

In the following table we can see the valuation in the calculations for varying priors. We can see that, as expected, Participation increases with the prior. However,  ceteris paribus, in some cases,  an increase in the prior may induce a decrease in the correctness. Intuitively, this occurs  as  the public signal weight increases in the players' contribution decision.
\FloatBarrier
\begin{table}[]
	\centering
	\caption{Correctness and Participation of $\Gamma(n=100,B=98,\mu,p=0.75,\tau)$}
	\label{tab:good_appx}
	\begin{tabular}{|c|c|c|c|c|c|c|c|c|}
		\hline
		$n$  & $B$ & $\mu$ & $p$  & $\tau$ & $\theta(n,B,\mu,p,\tau)$ & $R(n,B,\mu,p,\tau)$ & $\theta(\mu,p,\tau)$ & $R(\mu,p,\tau)$ \\ \hline
		100  & 50  & 0.2   & 0.75 & 0.5    & 0.946                    & 0.189               & 0.933                & 0.267           \\ \hline
		100  & 50  & 0.5   & 0.75 & 0.5    & 0.852                    & 0.489               & 0.833                & 0.667           \\ \hline
		100  & 50  & 0.7   & 0.75 & 0.5    & 0.776                    & 0.713               & 0.767                & 0.933           \\ \hline
		100  & 50  & 0.2   & 0.75 & 0.7    & 0.978                    & 0.171               & 0.971                & 0.229           \\ \hline
		100  & 50  & 0.5   & 0.75 & 0.7    & 0.940                    & 0.437               & 0.929                & 0.571           \\ \hline
		100  & 50  & 0.7   & 0.75 & 0.7    & 0.911                    & 0.623               & 0.900                & 0.700           \\ \hline \hline
		1000 & 500 & 0.2   & 0.75 & 0.5    & 0.937                    & 0.197               & 0.933                & 0.267           \\ \hline
		1000 & 500 & 0.5   & 0.75 & 0.5    & 0.839                    & 0.497               & 0.833                & 0.667           \\ \hline
		1000 & 500 & 0.7   & 0.75 & 0.5    & 0.769                    & 0.705               & 0.767                & 0.933           \\ \hline
		1000 & 500 & 0.2   & 0.75 & 0.7    & 0.973                    & 0.178               & 0.971                & 0.229           \\ \hline
		1000 & 500 & 0.5   & 0.75 & 0.7    & 0.932                    & 0.448               & 0.929                & 0.571           \\ \hline
		1000 & 500 & 0.7   & 0.75 & 0.7    & 0.903                    & 0.631               & 0.900                & 0.700           \\ \hline
	\end{tabular}
\end{table}

\FloatBarrier

%
%From Table \ref{tab:tab1}, we can see that even when $n=100,$ our results for large Crowdfunding games provide good approximation for player strategies and as a result to both indecies for campaign efficiency defined by \cite{Arieli2017}.  Our extended set of parameters uncovers new equilibria missing from previous analysis.  For example, when $\tau=0.7$ and the quality of signal is $p=0.55,$ low therhold induces an equilibrium is in which `low' players opt-out and `high' players mix.
%
%
%Naturally as prices increase, players inclination to contribute diminishes thus decreasing the probability that a threshold is reached. From Table \ref{tab:tab1}, one can see that by allowing price variation, a campaign may achieve maximal participation when setting the threshold around $1-p,$ eliciting participation solely from `high' players.  In our asymptotic results we present the conditions under which this occurs for large campaigns as well.
%
\section{Discussion}

In this paper we report theoretical findings about crowdfunding campaigns - strategies, correctness and participation - for large markets (presented in terms of asymptotic results). We then go on to compute outcomes in campaigns where the market size is inspired by empirical findings using field data. The contribution of the computational part is in showing that the theory holds even for markets of small size.

Our model supports variations in the product price and in the prior belief that the product is viable.
The study of asymmetric prior beliefs is of importance as a typical crowdfunding scenario is that of a high-risk product in its pre-development stage. To model this one should consider a low prior for the state of the world where the product is valuable.  Another type of risk is manifested in the  price variation. Expensively priced product embed a greater loss if they are not viable and smaller gains if they are.

When examining the various cases, we found that the risk associated with a high price are, in a sense,  more instrumental for the analysis then the risk conveyed in a low prior.
When prices are sufficiently low we find out that crowdfunding campaigns do not provide any value in terms of sieving out the bad products from the good ones. In fact, for such prices, as intuition suggests, participation is maximal and the correctness index is equal the prior probability.

Typically, the maximal participation decreases as the signal quality improves. However, this is no longer true for sufficiently high prices. In such cases we notice that the correctness index approaches one as the price approaches the maximal value of one.

In our model we show that whenever the product is bad the expected number of contributions roughly equals the campaign's threshold. Thus, for risky products, ones that exhibit high risk, the unconditional participation index is almost equal the threshold.  This is indeed observed in field data as reported in \cite{Kuppuswamy2018} and \cite{Mollick2014}.

A crucial primitive of our model is the information structure, composed of the initial common prior regarding the product quality and the accuracy of the signals available to each of the agents. Our analysis allows for comparative statics and shows how the two indices behave as functions of the informational primitives. Holding the price fixed our model predicts that the aforementioned Correctness index decreases the as the product becomes more risky (a lower prior for the good state). In addition, higher prices induce higher correctness as they decrease players' expected utility from contributing.

\section{Concluding remarks}\label{sec:conc}
Crowdfunding is often used by many entrepreneurs to validate the market demand for innovative products or an art project. We study how well do Crowdfunding campaigns perform in this context. To do so we introduce a vary simple game of incomplete information which we call \textit{the Crowdfunding game}.  We consider two success measures for a Crowdfunding game. First, the `Correctness index' of a game which captures how well information is aggregated, and second, the `Participation index' that reflects how convincing the campaign is. We show that for large populations information is not fully aggregated and we provide bounds on the correctness and penetration index for large populations.

Our results are primarily asymptotic. However, calculations show that these asymptotic bounds provide good approximations for realistic values of populations size, sometimes as small $10$ players. In fact, even when the number of players is finite, all three parameters we measure: $\sigma^*,$ the Correctness index $\theta$ and the Participation index $R$, are quite close to their respective asymptotic values. This observation is robust to the game parameters $\Params$.

\subsection{Future directions of research}

The static model we study is quite elementary and a few natural extensions that could possibly change some of the qualitative results are called for.
\begin{itemize}
\item
Realistically, the value of many goods has a private component and so we would like to study how crowdfunding performs in an environment where the value of the good has some private component, in addition to the common component.
\item
In most campaigns firms offer a menu of bundles (or variants) and prices. In our model we reduced this to a single product. We would like to verify that our reduction is not critical for the qualitative observations that we have \item

\end{itemize}

Finally, as already mentioned in the introduction, most campaigns take place over a period of time and a dynamic model may be called for. We suspect that in such a model the equilibrium analysis will show that most (if not all) the activity takes place in the final stage, in which case our static model serves as a meaningful approximation. Obviously, nothing guarantees that and a reasonable hypothesis is that more knowledgeable players will tend to move earlier as well as more optimistic players.

% Bibliography
\bibliographystyle{plain}
\bibliography{crowdfunding}
\newpage
% Appendix
\appendix

\section{Missing Proofs}\label{sec:proofs}

Before stating our auxiliary lemmas and proofs let us recap some of the relevant notation and introduce some new notation. Given a crowdfunding game, $\Gamma=\Gamma(n,B,\mu,p,\tau)$ all the notation refers to its unique non-trivial symmetric equilibrium strategy $\sigma^*$ and all the random variables pertain to the distribution, $Pr_{\Gamma,\sigma^*}$ over $\Omega\times S^n$,  given by the fundamentals of the game and $\sigma^*$.

\begin{itemize}
\item
For $s\in\{l,h\},$ let  $\sigma^*(s_i)=Pr(a_i=1|s_i)$, is the  probability in which a player with signal $s$ contributes.
\item
For any action vector $(a_1,\dots,a_n) \in \{o0,1\}^n$ let $c_n=\sum_{i=1}^{n}a_i$ be the number of  contributions.
\item
For any $\gamma\in[0,1]$  we let $z(\gamma,n)$ denote a binomial random variable with $n$ trials and a probability for success $\gamma$. Formally,  $z(\gamma,n)\sim Bin(\gamma,n).$
\item
For $k\in\{1\dots n\}$ define $\varphi(\gamma,n,k)= \pr(z(\gamma,n)\geq k).$ In words, it is the probability for $k$ successes or more in $n$ independent trials.
\item
The probability, according to a strategy $\sigma$,  to contribute in state $H$ is
$\alpha^{H}=\pr(a_i=1|\omega)=p\sigma(h)+(1-p)\sigma(l)$.
\item
The probability, according to a strategy $\sigma$,  to contribute in state $L$ is $\alpha^L=\pr(a_i=1|\omega=L)=(1-p)\sigma(h)+p\sigma(l).$
\item
The probability of a successful campaign in the equilibrium of the game $\Gamma(\Params)$, conditional on state $\omega$ is
$Pr(c_n\ge B | \omega) = Pr_{\Gamma,\sigma^*}(c_n\ge B | \omega)=\varphi(\alpha^{\omega},n,B).$
\end{itemize}

The utility of consumer $i$ with signal $s_i$ is:

\begin{eqnarray}\label{eq:Expected_util_gen}
E u_i(a_i=1|s_i)&=&\nonumber\\
& &\pr(\omega=H|s_i)\pr(c_{n-1}\ge B-1|\omega=H)(1-\tau)\\
&-&\pr(\omega=L|s_i)\pr(c_{n-1}\ge B-1|\omega=L)\tau=\nonumber\\
& & \pr(\omega=H|s_i)\varphi(\alpha^H,n-1,B-1)(1-\tau)-\pr(\omega=L|s_i)\varphi(\alpha^L,n-1,B-1)\tau.\nonumber
\end{eqnarray}

The player's expected utility induced by each signal is then:
\begin{eqnarray}\label{eq:Expected_util_gen2}
E u_i(a_i=1|h)&=&\nonumber\\ & & \frac{p\mu}{p\mu+(1-p)(1-\mu)}\varphi(\alpha^H,n-1,B-1)(1-\tau)\\ & &-\frac{(1-p)(1-\mu)}{p\mu+(1-p)(1-\mu)}\varphi(\alpha^L,n-1,B-1)\tau\nonumber\\
E u_i(a_i=1|l)&=&\nonumber\\ & & \frac{(1-p)\mu}{(1-p)\mu+p(1-\mu)}\varphi(\alpha^H,n-1,B-1)(1-\tau)\\
& &-\frac{p(1-\mu)}{(1-p)\mu+p(1-\mu)}\varphi(\alpha^L,n-1,B-1)\tau.\nonumber
\end{eqnarray}

\subsection*{Proof of Theorem \ref{thm:gen_atmost_one_eq}}

%
%In the following lemma we show that there are no symmetric, non-trivial, Bayes-Nash equilibria, in which both high and low players mix between the actions. In addition we show that in equilibrium, the probability of the threshold being achieved is higher (at least in the weak sense) for state $\omega=H$ then for $\omega=L$ and is strictly higher if $\sigma^*(l)<1.$

The following lemma specifies some characteristics of symmetric equilibria. Informally it suggests that the high type player is always more keen about contributing than the low type one.

\begin{lemma}\label{lem:psi_ge_lambda}
	Let $\sigma^*$ is a symmetric non-trivial Bayesian-Nash equilibrium in $\Gamma,$ then
	\begin{equation*}
	\begin{split}
	\mbox{ if } \sigma^*(l)>0&\mbox{ then }\sigma^*(h)=1\\
	\mbox{ if } \sigma^*(h)<1&\mbox{ then }\sigma^*(l)=0\\
	\end{split}
	\end{equation*}
	and
	\begin{equation*}
	\begin{split}
	&\mbox{ if } \sigma^*(l)<1\mbox{ then }\varphi(\alpha^H,n-1,B-1)>\varphi(\alpha^L,n-1,B-1)\\
	&\mbox{ if } \sigma(l)=1\mbox{ then }\varphi(\alpha^H,n-1,B-1)=\varphi(\alpha^L,n-1,B-1)=1.\\
	\end{split}
	\end{equation*}
\end{lemma}

\begin{proof}
	First assume that in equilibrium, the low player contributes with positive probability, that is $\sigma^*(l)>0.$ By equation \eqref{eq:Expected_util_gen2}, the following condition for the low type player must be satisfied,
	$$
	E_{\Gamma,\sigma^*} u_i(a_i=1|s_i=l)\ge 0 \Leftrightarrow \mu(1-p)\varphi(\alpha^H,n-1,B-1)(1-\tau)\ge(1-\mu)p\varphi(\alpha^L,n-1,B-1)\tau.
	$$
	As $p>\frac{1}{2}$ this entails that
%	$$
%	\mu p\varphi(\alpha^H,n-1,B-1)(1-\tau)>\mu(1-p)\varphi(\alpha^L,n-1,B-1)(1-\tau)
%	$$
%	and
%	$$
%	(1-\mu) p\varphi(\alpha^H,n-1,B-1)\tau>(1-\mu)(1-p)\varphi(\alpha^L,n-1,B-1)\tau
%	$$ and hence
	$$
	\mu p\varphi(\alpha^H,n-1,B-1)(1-\tau)> (1-\mu)(1-p)\varphi(\alpha^L,n-1,B-1)\tau.
	$$
	
	Note that if the above condition holds, the high player's expected utility from contributing is strictly positive by equation \eqref{eq:Expected_util_gen2}, therefore, contributing is her  best response when other players play $\sigma^*$ in the game $\Gamma(\Params).$
	
	Next assume that the high player assigns a positive probability to opting-out, that is $\sigma^*(h)<1.$ Then by equation \eqref{eq:Expected_util_gen2}, the following condition for the high type player must be satisfied,
	$$
	E u_i(a_i=1|s_i=h)\le 0 \Leftrightarrow\mu p\varphi(\alpha^H,n-1,B-1)(1-\tau)\le (1-\mu)(1-p)\varphi(\alpha^L,n-1,B-1)\tau.
	$$
	As $p>\frac{1}{2}$ this entails
	$$
	\mu (1-p)\varphi(\alpha^H,n-1,B-1)(1-\tau)< (1-\mu)p\varphi(\alpha^L,n-1,B-1)\tau.
	$$
	Note that by equation \eqref{eq:Expected_util_gen2}, whenever the condition above is satisfied, the utility of the low players is always negative, i.e, $E u_i(a_i=1|s_i=l)<0.$ Therefore, when all other players play $\sigma^*,$ the low players best response is to opt-out.

	Therefore if $\sigma(l)^*<1$ then $\sigma^*(h)>\sigma^*(l)$ and by  the definition of $\varphi(\cdot,\cdot,\cdot),\alpha^\omega$,  $\varphi(\alpha^H,n-1,B-1)>\varphi(\alpha^L,n-1,B-1),$
	and if $\sigma(l)=1$ then $\sigma(h)=1$ and $\varphi(\alpha^H,n-1,B-1)=\varphi(\alpha^L,n-1,B-1)=1.$
\end{proof}

%In Lemma \ref{lem:psi_ge_lambda} we present all possible ``types" of non-trivial, symmetric equilibria of our model.
Next we show that in any crowdfunding game there can be at most one symmetric non-trivial equilibrium.

%In the following propositions, we show that in any $\Gamma(\Params),$  there is at most one non-trivial, symmetric strategy $\sigma$ under which a player is indifferent.

\begin{proposition}\label{prop:low_indif_once}
	Let $\sigma$ be a non-trivial, symmetric strategy in $\Gamma.$
	If $\sigma(l)\in(0,1)$ and	$E_{\Gamma,\sigma} u_i(a_i=1|s_i=l)=0,$ then for any non-trivial, symmetric strategy $\tilde\sigma=(\tilde\sigma(l),\sigma(h))$ such that $\tilde\sigma(l)\in[0,\sigma(l)),$ $E_{\Gamma,\tilde\sigma} u_i(a_i=1|s_i=l)>0;$  and for any  non-trivial, symmetric strategy $\tilde\sigma=(\tilde\sigma(l),\sigma(h))$ such that $\tilde\sigma(l)\in(\sigma(l),1],$ $E_{\Gamma,\tilde\sigma} u_i(a_i=1|s_i=l)<0.$
\end{proposition}

\begin{proof}
	Fix the game parameters $p,\mu, n, B$ and $ \tau.$ Let $f$ be the following function on $[0,1]^2$: 
	%	For any $\psi,\lambda\in(0,1),$ We define  the function
%	\begingroup\makeatletter\def\f@size{8}\check@mathfonts
	\begin{eqnarray}\label{eq:f1}
	f(\lambda;\psi)&=& \nonumber\\
	& & \frac{(1-p)\mu}{(1-p)\mu+p(1-\mu)}\varphi(p\psi+(1-p)\lambda,n-1,B-1)(1-\tau)\\
	& & -(1-\frac{(1-p)\mu}{(1-p)\mu+p(1-\mu)})\varphi((1-p)\psi+p\lambda,n-1,B-1)\tau.\nonumber
	\end{eqnarray}
%	\endgroup
	
	Note that $f$ is the expected utility of a low player in $\Gamma$ when the players play a symmetric strategy $\sigma$ where $\sigma(h)=\psi$ and $\sigma(l)=\lambda.$ We fix the parameter $\psi.$	 
	 Let $\sigma$ be the strategy described in the proposition and let $\bar\psi$ denote the probability that the high player contributes. Now let $f(\lambda) = f(\lambda,\bar\psi)$ be a single parameter function.
	 
	 Recall  that $f(\lambda)$ is continuous in $\lambda,$ therefore, to prove that there is at most one value of $\lambda$ for which $f(\lambda)=0$ it suffices to prove that  whenever $f(\lambda)=0$ then $f'(\lambda)<0.$ 
	 \blue{Moran: Is it clearer now?}

	Assume to the contrary that there exists $\lambda\in(0,1)$ where $f(\lambda)=0$ and  $f'(\lambda)\ge 0.$
	
	Taking the derivative of $f$ entails
%	\begingroup\makeatletter\def\f@size{8}\check@mathfonts
	\begin{eqnarray}
	f'(\lambda)&=&\nonumber\\
	& &\frac{(1-p)\mu}{(1-p)\mu+p(1-\mu)}(1-p)\varphi'(p\psi+(1-p)\lambda,n-1,B-1)(1-\tau)\nonumber\\
	& &-(1-\frac{(1-p)\mu}{(1-p)\mu+p(1-\mu)})p\varphi'((1-p)\psi+p\lambda,n-1,B-1)\tau.\nonumber
	\end{eqnarray}
%	\endgroup
	As $p>\frac{1}{2},$
%	\begingroup\makeatletter\def\f@size{8}\check@mathfonts
	\begin{eqnarray}\label{eq:20}
	0<f'(\lambda)&<&\\
	& & p\big(\frac{(1-p)\mu}{(1-p)\mu+p(1-\mu)}\varphi'(p\psi+(1-p)\lambda,n-1,B-1)(1-\tau)\nonumber\\
	& &  -(1-\frac{(1-p)\mu}{(1-p)\mu+p(1-\mu)})\varphi'((1-p)\psi+p\lambda,n-1.B-1)\tau\big).\nonumber
	\end{eqnarray}
%	\endgroup
	By a standard argument  (see for example Feller \cite{Feller1968}, pp. 173), for any $\rho\in(0,1)$
	$\varphi'(\rho,n-1,B-1)=(n-1)\binom{n-2}{B-2}\rho^{B-2}(1-\rho)^{n-B}$ therefore, by equation \eqref{eq:20}, the assumption $f'(\lambda)\ge 0$ entails the following condition,
	$$0\le f'(\lambda)<C\eta(n-1,B-2)$$ where
	$C=(n-1)\binom{n-2}{B-2}p$ and
	\begin{eqnarray}
	\eta(k,n)&=&\\
	& & \frac{(1-p)\mu}{(1-p)\mu+p(1-\mu)}(p\psi+(1-p)\lambda)^k(1-(p\psi+(1-p)\lambda))^{n-k}(1-\tau)-\nonumber\\
	& & (1-\frac{(1-p)\mu}{(1-p)\mu+p(1-\mu)})(1-p)\psi+p\lambda)^{k}(1-((1-p)\psi+p\lambda))^{n-k}\tau.
	\end{eqnarray}
	As $C>0$ and we assume $f'(\lambda)>0,$ it must be the case that $\eta(B-2,n-1)>0.$
	
	Next we show that $\eta(k,n-1)>0$ yields that $\eta(k+1,n-1)>0.$ To see this note,
	\begingroup\makeatletter\def\f@size{9.5}\check@mathfonts
	\begin{eqnarray*}
	\eta(k+1,n-1)&=&\\
	& &\frac{(1-p)\mu}{(1-p)\mu+p(1-\mu)}(p\psi+(1-p)\lambda)^{k+1}(1-(p\psi(1-p)\lambda))^{n-1-(k+1)}(1-\tau)\\
	&-& (1-\frac{(1-p)\mu}{(1-p)\mu+p(1-\mu)})((1-p\psi+p\lambda))^{k+1}(1-((1-p)\psi+p\lambda))^{n-1-(k+1)}\tau\\
	&=&\frac{(p\psi+(1-p)\lambda)}{1-(p\psi+(1-p)\lambda)}\frac{(1-p)\mu}{(1-p)\mu+p(1-\mu)}(p\psi+(1-p)\lambda)^k(1-(p\psi+(1-p)\lambda))^{n-1-k}(1-\tau)\\
	&-&\frac{(1-p)\psi+p\lambda}{1-((1-p)\psi+p\lambda)}(1-\frac{(1-p)\mu}{(1-p)\mu+p(1-\mu)})((1-p)\psi+p\lambda)^{k}(1-((1-p)\psi+p\lambda))^{n-1-k}\tau\\
	& \ge& \frac{(1-p)\psi+p\lambda}{1-((1-p)\psi+p\lambda)}\eta(k,n-1).
	\end{eqnarray*}
	\endgroup
	Where the last inequality holds as  $p>\frac{1}{2},\psi\ge\lambda$ thus 
	$$p\psi+(1-p)\lambda>(1-p)\psi+p\lambda$$
	and the function $\frac{x}{1-x}$ increases monotonically whenever $x\in(0,1).$
	
	The assumption $f'(\lambda)\ge 0$ entails that $\eta(B-2,n-1)>0$ .
	By definition of Binomial distribution (see for example  \cite{Feller1968}, pp. 147)
	$f(\lambda)=\sum_{k=B-1}^{n-1}\binom{n-1}{k}\eta(k,n-1)>\sum_{k=B-1}^{n-1}\eta(B-2,n-1)>0.$
	This is in contradiction with our assumption that $f(\lambda)=0$.
\end{proof}

\begin{proposition}\label{prop:high_indif_once}
	Let $\sigma$ be a non-trivial, symmetric strategy in $\Gamma.$  If $\sigma(h)\in(0,1)$ and	$E_{\Gamma,\sigma} u_i(a_i=1|s_i=h)=0,$ then for any non-trivial, symmetric strategy $\tilde\sigma=(\sigma(l),\tilde\sigma(h))$ such that $\tilde\sigma(h)\in[0,\sigma(h)),$ $E_{\Gamma,\tilde\sigma} u_i(a_i=1|s_i=h)>0;$  and for any  non-trivial, symmetric strategy $\sigma=(\sigma(l),\tilde\sigma(h))$ such that $\tilde\sigma(h)\in(\sigma(h),1],$ $E_{\Gamma,\tilde\sigma} u_i(a_i=1|s_i=h)<0.$
\end{proposition}

\begin{proof}
	The proof of Proposition \ref{prop:high_indif_once} is very similar to that of Proposition \ref{prop:low_indif_once}.
	 Fix the game parameters $p,\mu,n,B$ and $\tau.$ Let $\hat{f}$ be the following function on $[0,1]^2:$$\psi,\lambda\in(0,1),\psi\ge\lambda$ We define  the function
%	\begingroup\makeatletter\def\f@size{8}\check@mathfonts
	\begin{eqnarray}\label{eq:f}
	\hat{f}(\psi;\lambda)&=&\nonumber\\
	& & \frac{p\mu}{p\mu+(1-p)(1-\mu)}\varphi(p\psi+(1-p)\lambda,n-1,B-1)(1-\tau)\\
	&-&(1-\frac{p\mu}{p\mu+(1-p)(1-\mu)})\varphi((1-p)\psi+p\lambda,n-1,B-1)\tau.\nonumber
	\end{eqnarray}
%	\endgroup
	
	Note that $\hat{f}$ is the expected utility of a high player in $\Gamma$ when the players play a symmetric strategy $\sigma$ where $\sigma(h)=\psi$ and $\sigma(l)=\lambda.$ We fix the parameter $\lambda.$ Let $\sigma$ be the strategy described in the proposition and let $\bar\lambda$ denote the probability that the high player contributes. Now let $\hat{f}(\psi) = \hat{f}(\bar\lambda,\psi)$ be a single parameter function.
	
	Recall  that $\hat{f}(\psi)$ is continuous in $\psi,$ therefore, to prove that there is at most one value of $\psi$ for which $\hat{f}(\psi)=0$ it suffices to prove that  whenever $\hat{f}(\psi)=0$ then $\hat{f}'(\psi)<0.$ 
	\blue{Moran: Is it clearer now?}
	
	Assume to the contrary that there exists $\psi\in(0,1)$ where $\hat{f}(\psi)=0$ and  $\hat{f}'(\psi)\ge 0.$
	
	Taking the derivative of $\hat{f}$ entails
	%\begingroup\makeatletter\def\f@size{8}\check@mathfonts
	\begin{eqnarray}
	\hat{f}'(\psi)&=&\nonumber\\
	& &\frac{p\mu}{p\mu+(1-p)(1-\mu)}p\varphi'(p\psi+(1-p)\lambda,n-1,B-1)(1-\tau)\\ 
	& &-(1-\frac{p\mu}{p\mu+(1-p)(1-\mu)})(1-p)\varphi'((1-p)\psi+p\lambda,n-1,B-1)\tau.\nonumber
	\end{eqnarray}
	%\endgroup
	As $p<1,$
	%\begingroup\makeatletter\def\f@size{7}\check@mathfonts
	\begin{eqnarray*}		
	0<\hat{f}'(\psi)< &\frac{p}{1-p}\big(\frac{p\mu}{p\mu+(1-p)(1-\mu)}\varphi'(p\psi+(1-p)\lambda,n-1,B-1)(1-\tau)&\\
	&  -(1-\frac{p\mu}{p\mu+(1-p)(1-\mu)})\varphi'((1-p)\psi+p\lambda,n-1.B-1)\tau\big).&
	\end{eqnarray*}
	%\endgroup
	By a standard argument (see for example Feller \cite{Feller1968} pp.173), for any $\rho\in(0,1)$
	$\varphi'(\rho,n-1,B-1)=(n-1)\binom{n-2}{B-2}\rho^{B-2}(1-\rho)^{n-B}$ therefore, the condition above entails,
	$0\le \hat{f}'(\lambda)<\hat{C}\hat{\eta}(n-1,B-2)$ where
	$\hat{C}=(n-1)\binom{n-2}{B-2}\frac{p}{1-p}$ and
	\begin{eqnarray*}
	 \hat{\eta}(k,n)=&\frac{p\mu}{p\mu+(1-p)(1-\mu)}(p\psi+(1-p)\lambda)^k(1-(p\psi+(1-p)\lambda))^{n-k}(1-\tau)&\\
	 &-(1-\frac{p\mu}{p\mu+(1-p)(1-\mu)})(1-p)\psi+p\lambda)^{k}(1-((1-p)\psi+p\lambda))^{n-k}\tau.&\end{eqnarray*}
	As $\hat{C}>0$ and we assume $\hat{f}'(\psi)>0,$ it must be the case that $\eta(B-2,n-1)>0.$
	
	Next we show that $\hat{\eta}(k,n-1)>0$ yields that $\hat{\eta}(k+1,n-1)>0$ as well.
	%\begingroup\makeatletter\def\f@size{5}\check@mathfonts
	\begin{eqnarray*}
	\hat{\eta}(k+1,n-1)=
	 \frac{p\mu}{p\mu+(1-p)(1-\mu)}\big(p\psi+(1-p)\lambda)^{k+1}(1-(p\psi(1-p)\lambda)\big)^{n-1-(k+1)}(1-\tau)& & \\-(1-\frac{p\mu}{p\mu+(1-p)(1-\mu)})((1-(p\psi+p\lambda))^{k+1}(1-((1-p)\psi+p\lambda))^{n-1-(k+1)}\tau=& &\\
	 \frac{(p\psi+(1-p)\lambda)}{1-(p\psi+(1-p)\lambda)}\frac{p\mu}{p\mu+(1-p)(1-\mu)}(p\psi+(1-p)\lambda)^k(1-(p\psi+(1-p)\lambda))^{n-1-k}(1-\tau)&& \\
	 -\frac{p\mu}{p\mu+(1-p)(1-\mu)}(1-\frac{(1-p)\mu}{(1-p)\mu+p(1-\mu)})((1-p)\psi+p\lambda)^{k}(1-((1-p)\psi+p\lambda))^{n-1-k}\tau& &\\
	\ge \frac{(1-p)\psi+p\lambda}{1-((1-p)\psi+p\lambda)}\hat{\eta}(k,n-1).& &
	\end{eqnarray*}
	Where the last inequality holds as  $p>\frac{1}{2},\psi\ge\lambda$ thus 
$$p\psi+(1-p)\lambda>(1-p)\psi+p\lambda$$
and the function $\frac{x}{1-x}$ increases monotonically whenever $x\in(0,1).$

The assumption $\hat{f}'(\lambda)\ge 0$ entails that $\hat{\eta}(B-2,n-1)>0$ .
By definition of Binomial distribution (see for example  \cite{Feller1968}, pp. 147)
$\hat{f}(\lambda)=\sum_{k=B-1}^{n-1}\binom{n-1}{k}\hat{\eta}(k,n-1)>\sum_{k=B-1}^{n-1}\hat{\eta}(B-2,n-1)>0.$
This is in contradiction with our assumption that $\hat{f}(\lambda)=0$.
\end{proof}

\begin{lemma}\label{lem:gen_atmost_one_eq-p1}
No crowdfunding game has more than one symmetric non-trivial Bayes-Nash equilibrium.	
\end{lemma}

\begin{proof}
We prove the first part of the theorem, that is for any $\Params$ there can be at most a single non-trivial symmetric Bayes-Nash equilibrium in $\Gamma(n,B,\mu,p,\tau).$
Let $\sigma$ be a symmetric, non-trivial strategy of $\Gamma(\Params).$ We separate the proof into cases and search for strategies that are candidates for equilbria.

\textbf{Case 1. $\sigma(h)=1$:}  Consider the following  sub-cases: (1.1) First assume that $$E_{\sigma=(\sigma(l)=\lambda,\sigma(h)=1)} u_i(a_i=1|s_i=l)= 0$$ for some $\lambda\in(0,1).$ 
Note that if the other low players play $\hat{\lambda}<\lambda,$ then by Proposition \ref{prop:low_indif_once}, $E_{\sigma=(\hat\lambda,1)} u_i(a_i=1|s_i=l)>0.$ In this case a low player's best response is to contribute. For  any strategy in which $\hat{\lambda}>\lambda,$ again by Proposition \ref{prop:low_indif_once}, $E_{\sigma=(\hat\lambda,1)} u_i(a_i=1|s_i=l)<0$ and opting-our is a low player's best response. However if $\sigma(l)=\lambda,$ then a low type player is indifferent between the actions and can not profit from increasing the probability she assigns to any one of the pure actions.  In addition, by Lemma \ref{lem:psi_ge_lambda}, if low players mix, then the high player has a dominant strategy $\sigma(h)=1,$ therefore, in this case, the only non-trivial, symmetric equilibrium can be $\sigma=(\sigma(l)=\lambda,\sigma(h=1)).$

(1.2) Assume that $$E_{\sigma=(\lambda,1)} u_i(a_i=1|s_i=l)\ne 0$$ for every $\lambda\in[0,1].$ By the continuity of $E_{\sigma=(\lambda,1)} u_i(a_i=1|s_i=l)$  it must be that $E_{\Gamma,\sigma=(\lambda,1)} u_i(a_i=1|s_i=l)>0$ or $E_{\sigma=(\lambda,1)} u_i(a_i=1|s_i=l)<0$ for every $\lambda.$ If $E_{\sigma=(\lambda,1)} u_i(a_i=1|s_i=l)>0$ for every $\lambda\in[0,1]$ then when $\sigma(h)=1,$ contributing is a best response strategy for a low player  and the only equilibrium of this form can only be $\sigma(l)=\sigma(h)=1.$
If $E_{\sigma=(\lambda,1)} u_i(a_i=1|s_i=l)<0$ for every $\lambda.$ then opting-out is a dominant action for a low type player.  If $E_{\sigma=(0,1)} u_i(a_i=1|s_i=h)\ge 0,$ then, by Proposition  \ref{prop:high_indif_once},   $\sigma=\sigma(l)=0,\sigma(h)=1$ is the only equilibrium candidate.
If $E_{\sigma=(0,1)} u_i(a_i=1|s_i=h)< 0,$ then we claim that opting-out is a dominant strategy for low players.
To see this, assume to the contrary that $E_{\sigma=(0,1)} u_i(a_i=1|s_i=h)< 0,$ and there exists some $\lambda\in(0,1]$ such that $E_{\sigma=(\lambda,1)} u_i(a_i=1|s_i=l)\ge 0$, which, by Proposition \ref{prop:low_indif_once} entails that  $E_{\sigma=(0,1)} u_i(a_i=1|s_i=l)> 0.$ a contradiction as $E_{\sigma=(0,1)} u_i(a_i=1|s_i=h)>E_{\sigma=(0,1)} u_i(a_i=1|s_i=l).$

\textbf{Case 2. $\sigma(h)<1:$} By Lemma \ref{lem:psi_ge_lambda}, the only candidates for non-trivial symmetric equilibria can strategies in which $\sigma(l)=0.$ We separate the proof to two sub-cases. (2.1) Assume that  $$E_{\sigma=(\sigma(l)=0,\sigma(h)=\psi)} u_i(a_i=1|s_i=h)= 0$$ for some $\psi\in(0,1].$ Then, by Proposition \ref{prop:high_indif_once}, for every $\tilde{\psi}\in[0,\psi),$ $E_{\sigma=(\sigma(l)=0,\sigma(h)=\tilde\psi)} u_i(a_i=1|s_i=h)> 0$ therefore high player is best if she contributes when all other players play $\sigma=(\sigma(l)=0,\sigma(h)=\tilde\psi)$  and for every $\tilde{\psi}\in(\psi,1]$  $E_{\sigma=(\sigma(l)=0,\sigma(h)=\tilde\psi)} u_i(a_i=1|s_i=h)< 0$ therefore high player is best if she opts-out. The high players is indifferent, only when $\sigma(h)=\psi,$ therefore she can not profit by deviating. Note that as $E_{\sigma=(\sigma(l)=0,\sigma(h)=\psi)} u_i(a_i=1|s_i=h) = 0,$ low player's best response is $\sigma(l)=0.$ 
(2.2)  Assume that  $$E_{\sigma=(\sigma(l)=0,\sigma(h)=\psi)} u_i(a_i=1|s_i=h)\ne 0$$ for every $\psi\in(0,1).$ By the continuity of $E_{\sigma=(\sigma(l)=0,\sigma(h)=\psi)} u_i(a_i=1|s_i=h),$ A non-trivial symmetric equilibrium can only occur if for every $\psi,$ $E_{\sigma=(\sigma(l)=0,\sigma(h)=\psi)} u_i(a_i=1|s_i=h) > 0,$ In this case we get that an equilibrium can occur only when $\sigma(h)=1.$ This case was analyzed in case (1) above.

We conclude, If $E_{\sigma=(0,1)}u_i(a=1|s_i=h)>0$ then $\sigma^*(h)=1.$ In addition: if $E_{\sigma=(0,1)}u_i(a=1|s_i=l)<0$ then $\sigma^*(l)=0;$ if $E_{\sigma=(1,1)}u_i(a=1|s_i=l)\ge 0$ then $\sigma^*(l)=1;$ Otherwise, by Lemma \ref{prop:low_indif_once}, there exist $\lambda\in(0,1)$ for which $E_{\sigma=(\lambda,1)}u_i(a=1|s_i=l)=0$ and $\sigma^*(l)=\lambda.$

Else if $E_{\sigma=(0,1)}u_i(a=1|s_i=h)<0$ then $\sigma^*(l)=0.$ In addition if $E_{\sigma=(0,0)}u_i(a=1|s_i=h)>0$  then by Proposition \ref{prop:high_indif_once}, there exist $\psi\in(0,1)$ such that $E_{\sigma=(0,\psi)}u_i(a=1|s_i=h)=0.$ By Proposition \ref{prop:high_indif_once} again, $E_{\sigma=(0,\tilde\psi)}u_i(a=1|s_i=l)<0$ for any $\tilde{\psi}\in[0,\psi),$ therefore if a high players plays the strategy $\sigma(h)=\tilde{\psi},$ any high player can profit by deviating to action $1$ and  similarly if high players play $\sigma(h)=\tilde{\psi}\in(\psi,1],$ any high player can profit by deviating to action $0.$ When high players play $\sigma(h)=\psi,$ high players have no profitable deviation. In addition, as  $E_{\sigma=(0,\psi)}u_i(a=1|s_i=h)=0>E_{\sigma=(0,\psi)}u_i(a=1|s_i=l),$ low players can not gain by deviating to action $1$ and thus $\sigma^*(l)=0,\sigma^*(h)=\psi$ is an equiblirium. Finally, if $E_{\sigma=(0,0)}u_i(a=1|s_i=h)\le 0$ then, there is no non-trivial symmetric equilibrium of $\Gamma(\Params).$
\end{proof}
%
%
%
%
%$a_i=0$ is dominant for high players as well, and thus the only symmetric equilibrium in $\Gamma(n,b,\mu,p,\tau)$ is the trivial equilibrium. If for every $\psi\in[0,1]$
%$E_{\sigma=(0,\psi)} u_i(a_i=1|s_i=h)> 0,$ then $a_1=1$ is a dominant strategy for high players therefore the only equilibrium is $\sigma(l)=0,\sigma(h)=1.$ Finally if there exist $\hat\psi\in(0,1)$  such that $E_{\sigma=(0,\hat\psi)} u_i(a_i=1|s_i=h)=0$ then by Corollary \ref{prop:high_indif_once}, for every $\psi\in[0,\hat\psi)$ $E_{\sigma=(0,\psi)} u_i(a_i=1|s_i=h)>0$ therefore $a_i=1$ is optimal for a high player and for every $\psi\in(\hat\psi,1]$ $E_{\sigma=(0,\psi)} u_i(a_i=1|s_i=h)<0$ therefore $a_i=0$ is optimal for a high player. When $\sigma(h)=\hat\psi$ high players are indifferent between both actions thus, can mix. Therefore the only equilibrium is $\sigma(l)=0,\sigma(h)=\hat\psi\in(0,1).$

To complete the proof we will show that for any $4-$tuple of parameters $(q,\mu,p,\tau)$ where $q\in(0,1)$, and for any there exists some $N$ such that for any $n>N,$ the crowdfunding game $\Gamma(n,B,\mu,p,\tau)$ has a unique non-trivial Bayes-Nash equilibrium.

Next we will prove the second part of Theorem \ref{thm:gen_atmost_one_eq}.

\begin{lemma}\label{lem:big_n_at_least_one}
	Let $q\in(0,1)$ and $\{B_n\}_{n=1}^\infty$ be a sequence of thresholds such that $\lim_{n\rightarrow\infty}\frac{B_n}{n}=q\in(0,1].$
	For every 3-tuple $\mu,q,\tau$ there exists $N(\mu,q,\tau)$ such that  every $n>N(\mu,q,\tau),$ there exist a unique non-trivial symmetric equilibrium of $\Gamma(n,B_n,\mu,p,\tau).$
\end{lemma}

\begin{proof}	
	Assume to the contrary that for some tuple $(\mu,q,\tau)$ and a sequence of thresholds $\{B_n\}_{n=1}^\infty$ for which $\lim_{n\rightarrow\infty}\frac{B_n}{n}=q,$  there exists an arbitrarily large $n$ such that the corresponding game $\Gamma(n,B_n,\mu,p,\tau)$ has no non-trivial symmetric equilibrium.
	
	This entails that for every $\psi\in(0,1],$ $$E_{\Gamma^n,\sigma(l)=0,\sigma(h)=\psi}u_i(a_i=1|s_i=h)<0.$$
	To see this, first consider a case in which the expected utility for high player is positive for every $\psi\in(0,1),$ in this case,  action $1$ is a dominant strategy for high players and thus a non-trivial symmetric equilibrium or the form  $\sigma^*(h)=1,\sigma^*(l)\in[0,1]$ exists; second consider a case in which there exists $\psi\in(0,1)$ for which $E_{\Gamma^n,\sigma(l)=0,\sigma(h)=\psi}u_i(a_i=1|s_i=h)=0.$ By Proposition \ref{prop:high_indif_once} there can be at most one such $\psi.$ Note that this yields that $\sigma^*(h)=\psi,\sigma^*(l)=0$ as no player can gain from deviating to any of the pure actions.
	
	Therefore, by the contrary assumption, for any $\psi\in(0,1]$ following condition holds,
	\begin{eqnarray}\label{eq:cond1}
	\frac{p\mu}{p\mu+(1-p)(1-\mu)}\varphi(p\psi,n-1,B_n-1)(1-\tau)-& &\nonumber\\
	\frac{(1-p)(1-\mu)}{p\mu+(1-p)(1-\mu)}\varphi((1-p)\psi,n-1,B_n-1)\tau<0\Leftrightarrow& &\nonumber\\
	\frac{p\mu}{p\mu+(1-p)(1-\mu)}\varphi(p\psi,n-1,B_n-1)(1-\tau)
	<& &\\
	\frac{(1-p)(1-\mu)}{p\mu+(1-p)(1-\mu)}\varphi((1-p)\psi,n-1,B_n-1)\tau \Leftrightarrow& &\nonumber\\
	\frac{p\mu(1-\tau)}{(1-p)(1-\mu)\tau}<
	\frac{\varphi((1-p)\psi,n-1,B_n-1)}{\varphi(p\psi,n-1,B_n-1)}.& &\nonumber
	\end{eqnarray}
	As stated above, equation \eqref{eq:cond1} must be satisfied for any $\psi\in(0,1],$ including arbitrarily small values and thus,
	\begin{equation}\label{eq:cond2}
	\begin{split}
	&\frac{p\mu(1-\tau)}{(1-p)(1-\mu)\tau}<\lim_{\psi\rightarrow 0}\frac{\varphi((1-p)\psi,n-1,B_n-1)}{\varphi(p\psi,n-1,B_n-1)}.
	\end{split}
	\end{equation}
	By the definition of $\varphi(\cdot,\cdot,\cdot),$
	$$\lim_{\psi=0}\varphi((1-p)\psi,n-1,B_n-1)=\lim_{\psi\rightarrow 0}\varphi((1-p)\psi,n-1,B_n-1)=0,$$
	and hence we must apply L'hopital's rule to calculate the limit of the right-hand side of equation \eqref{eq:cond2}, that is,
	\begin{equation}\label{eq:cond3}
	\begin{split}
	&\frac{p\mu(1-\tau)}{(1-p)(1-\mu)\tau}<\lim_{\psi\rightarrow 0}\frac{\varphi'((1-p)\psi,n-1,B_n-1)}{\varphi'(p\psi,n-1,B_n-1)}.
	\end{split}
	\end{equation}
	
	By a standard argument (see for example Feller \cite{Feller1968} pp.173), 
	\begin{equation}
	\varphi(\gamma,n-1,B-1)=(n-1)\binom{n-2}{B-2}\int_0^{\gamma} t^{B-2}(1-t)^{n-B}dt.
	\end{equation}
	and hence
	\begin{equation}\label{eq:der_varphi_1}
	\varphi'(\gamma,n-1,B-1)=(n-1)\binom{n-2}{B-2}\gamma^{B-2}(1-\gamma)^{n-B}.
	\end{equation}
	Therefore
	\begin{equation}\label{eq:psi_limit_zero}
	\begin{split}
	&\lim_{\psi\rightarrow 0} \frac{\varphi((1-p)\psi,n-1,B-1)}{\varphi(p\psi,n-1,B-1)}=
	\lim_{\psi\rightarrow 0} \frac{\varphi'((1-p)\psi,n-1,B-1)(1-p)}{\varphi'(p\psi,n-1,B-1)p}=\\
	&\lim_{\psi\rightarrow 0}(\frac{(1-p)}{p})^{B-1}(\frac{1-(1-p)\psi}{1-p\psi})^{n-B}=\\
	&(\frac{1-p}{p})^{B-1}.
	\end{split}
	\end{equation}
	
	We plug this into equation \eqref{eq:cond3} and get,
	
	\begin{equation}\label{eq:cond3}
	\begin{split}
	&\frac{p\mu(1-\tau)}{(1-p)(1-\mu)\tau}<(\frac{1-p}{p})^{B_n-1},
	\end{split}
	\end{equation}
	a contradiction as $p>\frac{1}{2}$ and $n$ is arbitrarily large (and thus so is $B_n$).
\end{proof}

Theorem \ref{thm:gen_atmost_one_eq} joins together Lemmas \ref{lem:big_n_at_least_one} and \ref{lem:gen_atmost_one_eq-p1}.

 \begin{customthm}{\ref{thm:gen_atmost_one_eq}}
 	(1) No crowdfunding game has more than one symmetric non-trivial Bayes-Nash equilibrium. (2) Consider the sequence $\{B_n\}_{n=1}^{\infty}$ where $\lim_{n\rightarrow\infty}\frac{B_n}{n}=q$ for some $q\in(0,1].$   For any 4-tuple of parameters $(q,\mu,p,\tau)$ there exists some $N$ such that for any $n>N,$  the crowdfunding game $\Gamma(n,B_n,\mu,p,\tau)$ has a unique symmetric non-trivial Bayes-Nash equilibrium.	
 \end{customthm}

\subsection{Additional Proofs.}
The following lemma is almost immediate by Lemma \ref{lem:psi_ge_lambda},

\begin{lemma}\label{lem:dominant_one}
	For every $s_i\in\{l,h\},$ if $\tau<Pr(\omega=H|s_i)$ then $\sigma^*(a_i=1|s_i)=1.$
\end{lemma}
\begin{proof}
	The expected utility for player $i$ from the action $a_i=1$ is
	$$
	Pr(\omega=H|s_i)(1-\tau)\varphi(\alpha^H,n-1,B-1)-(1-Pr_\mu(\omega=H|s_i))\tau \varphi(\alpha^L,n-1,B-1).
	$$
	By Lemma \ref{lem:psi_ge_lambda} we know that $\varphi(\alpha^H,n-1,B-1)\ge \varphi(\alpha^L,n-1,B-1)$ and as $\sigma^*$ is non-trivial we know that $\varphi(\alpha^H,n-1,B-1)>0$ therefore,
	\begin{equation*}
	\begin{split}
	&Pr(\omega=H|s_i)(1-\tau)\varphi(\alpha^H,n-1,B-1)-(1-Pr_\mu(\omega=H|s_i))\tau \varphi(\alpha^L,n-1,B-1)\ge\\ & (Pr_\mu(\omega=H|s_i)(1-\tau)-(1-Pr_\mu(\omega=H|s_i))\tau )\varphi(\alpha^H,n-1,B-1)>0.
	\end{split}
	\end{equation*}
	Where the last inequality holds as we assume $\tau<Pr(\omega=H|s_i).$
\end{proof}

By Lemma \ref{lem:dominant_one}, we can distinguish between three crowdfunding game types by the relationship between $\mu,$ the prior for state $H$  and the pre-determined price level $\tau.$ In Game $\Gamma(\Params),$ the price will be called \textit{cheap} if $\tau<Pr_\mu(\omega=H|s_i=l)=\frac{\mu(1-p)}{\mu(1-p)+p(1-\mu)};$ the price will be called \textit{moderate} if $\tau\in[Pr_\mu(\omega=H|s_i=l),Pr_\mu(\omega=H|s_i=l)=[\frac{(1-p)\mu}{(1-p)\mu+p(1-\mu)},\frac{p\mu}{p\mu+(1-p)(1-\mu)});$ and finally, the price will be called  \textit{expensive} if $\tau>\frac{p\mu}{p\mu+(1-p)(1-\mu)}.$

Let $\{B_n\}_{n=1}^{\infty}$ be a sequence of thresholds such that $\lim_{n\rightarrow\infty}\frac{B_n,n}=q$ for some $q\in(0,1).$ For every $\mu,p,\tau$ let $\Gamma_n$ denote the corresponding game $\Gamma(n,B_n,\mu,p,\tau).$ We denote the  non-trivial symmetric equilibria of $\Gamma_n$ by $\sigma^*_n.$ By by taking a sub-sequence if necessary, the limit $\lim_{n\rightarrow\infty}\sigma^*_n=\sigma^*_\infty,$ exists and is the non-trivial, symmetric equilibrium of game $\Gamma_\infty.$ By Theorem \ref{thm:gen_atmost_one_eq} we know this is well defined.

\begin{lemma}
	If $q\ge 1-p$ then for every $\mu,p,\tau,$ $$\sigma^*_\infty(h)=1.$$
\end{lemma}
\begin{proof}
	
	Assume by contradiction that there exist a sequence $\{B_n\}_{n=1}^{\infty}$ such that $\lim_{n\rightarrow\infty}\frac{B_n}{n}=q>1-p,$ and the corresponding sequence of equilibrium strategies $\{\sigma^*_n\}$ converges to $\sigma^*_\infty(h)=\psi$ for some $\psi\in(0,1).$ By Lemma \ref{lem:psi_ge_lambda}, if $\sigma^*_\infty(h)<1$ then  $\sigma^*_\infty(l)=0.$
	
	In addition, by the law of large numbers, the expected number of contributions $c^L_n\equiv\sum_i a_i$ in state $\omega=L$ is
	$$\lim_{n\rightarrow\infty}c^L_n= (1-p)\psi<q,$$
	and thus $\lim_{n\rightarrow\infty}\varphi(\alpha^L_n,n-1,B_n-1)=0.$ Therefore for sufficiently large $n,$ $E_{\sigma^*_n}u(a_i=1|s_i=l)>0,$ a contradiction as low player profits by deviating to the pure action $1.$
\end{proof}	

Lastly, we follow \cite{Roth2016} and \cite{Levine1995} and show that in our model, players are (almost) non-pivotal.
\begin{lemma}\label{lem:almost_unpivotal}
	Consider any tuple $p,\mu,tau,q$ and any sequence of thresholds $\{B_n\}$ such that $\limninf\frac{B_n}{n}=q.$ Let $\{\Gamma_n\},$ denote  the corresponding sequence of crowdfunding games and the corresponding equlibria sequence $\{\sigma^*_n\}$ the following equality must be satisfied,
	$$\limninf \varphi(\alpha^{\omega},n-1,B_n-1)=\limninf \varphi(\alpha^{\omega},n-1,B_n-1)$$
\end{lemma}
\begin{proof}
	Let $c^\omega_n$ denote the expected number of contributions. By definition of $\varphi$ we get,
	$$\varphi(\alpha^\omega,n-1,B_n-1)=Pr_{\Gamma_n,\sigma^*_n}(c^\omega_n\ge b_n-1|\omega)=Pr_{\Gamma_n,\sigma^*_n}(\frac{c^\omega_n-1}{n}\ge \frac{B_n}{n}|\omega).$$
	
	The above equality must also hold in the limit and thus,
	\begin{eqnarray*}
	\limninf\varphi(\alpha^\omega,n-1,B_n-1)&=&\limninf Pr_{\Gamma_n,\sigma^*_n}(\frac{c^\omega_n-1}{n}\ge \frac{B_n}{n}|\omega)=\\
	& &\limninf Pr_{\Gamma_n,\sigma^*_n}(\frac{c^\omega_n}{n}\ge \frac{B_n}{n}|\omega)=\\
	& &\limninf\varphi(\alpha^{\omega},n,B_n).
	\end{eqnarray*}
\end{proof}	

\subsection{Proofs for cheap prices.}\label{sec:cheap_price_proofs}

\begin{customthm}{3.1}
	In any crowdfunding game  with a cheap price there exists a unique symmetric Bayesian Nash equilibrium where all players contribute. In particular, this equilibrium is non-trivial.
\end{customthm}

\begin{proof}
	The proof of Theorem \ref{thm:eq_asym_game_towards} is immediate by Lemma \ref{lem:dominant_one} and the definition of cheap price.
\end{proof}

\subsection{Proofs for moderate prices.}\label{sec:cheap_price_proofs}

\begin{customthm}{3.3}
	For any crowdfunding game, $\Gamma(n,B,\mu,p,\tau)$, with a moderate price, there exists a \emph{unique} symmetric non-trivial Bayesian Nash equilibrium $\sigma^*=(\sigma^*_1,\ldots,\sigma^*_n).$ Moreover, $\sigma^*_i$ has the following form,
	\begin{equation}\label{eq_equilibrium}
	\sigma^*_i(s_i)=\begin{cases}
	1&\mbox{ if }s_i=h\\
	\lambda=\lambda(n,B,\mu,p,\tau))\in[0,1) &\mbox{ if }s_i=l.
	\end{cases}.
	\end{equation}
\end{customthm}

\begin{proof}
	By Lemma \ref{lem:dominant_one}, in any crowdfunding game $\Gamma$ with moderate price, $\sigma^*(h)=1.$

	First we show that $\sigma^*(l)<1.$ Let $\sigma_1$ be the strategy profile in which all players play the pure action $1.$ We use the following shorthand notation,
	$$p_l=Pr_{\Gamma}(\omega=H|s_i=l)=\frac{\mu(1-p)}{p(1-\mu)+\mu(1-p)}$$
	and
	$$p_h=Pr_{\Gamma}(\omega=H|s_i=h)=\frac{\mu p}{p\mu+(1-\mu)(1-p)}.$$	
	As all players contribute under $\sigma_1,$ the threshold is achieved in both states with probablity $1,$ and thus the expected utility of a low player is,
	$E_{\Gamma,\sigma_1} u_i(a_i=1|s_i=l)=p_l(1-\tau)-(1-p_l)\tau<0,$
	where the last inequality stems from the assumption of moderate prices. Therefore $\sigma^*(l)<1.$
	
	Next we show that there exists a non-trivial symmetric Bayes Nash equilibrium in any crowdfunding game with moderate price.
	We separate the proof into two cases.
	
	(1) For any $\lambda\in[0,1),$  $$E_{\Gamma,\sigma(h)=1,\sigma(l)=\lambda}u_i(a_i=1|s_i=l)<0.$$ In this case the action $0$ is a low player's best response against any symmetric strategy. Therefore $\sigma^*(l)=0.$	By Lemma \ref{lem:dominant_one}, $\sigma^*(h)=1$ and thus $\sigma^*(h)=1,\sigma^*(l)=0$ is an equilibrium. Note that by definition of $\varphi(cdot,cdot,cdot),$ the equilibrium is non-trivial for any $n,B,\mu,p$ and a moderate $\tau.$
	
	(2) There exists some $\lambda\in(0,1)$ such that $$E_{\Gamma,\sigma(h)=1,\sigma(l)=\lambda}u_i(a_i=1|s_i=l)=0.$$ By Proposition \ref{prop:low_indif_once} there can be at most one such $\lambda.$ We show that $\sigma^*(h)=1,\sigma^*(l)=\lambda$ is the unique non-trivial equilibrium of $\Gamma$ in this case. Note that $\sigma^*(h)=1$ by Lemma \ref{lem:dominant_one}. In addition, by Propositoin \ref{prop:low_indif_once}, the low type player loses utility by deviating to the pure action $1$ and can not gain by deviating to the pure action $0.$
\end{proof}

\begin{customlemma}{\ref{lem:asympt_probs}}
	Let $\{\Gamma(n,B_n,\mu,p,\tau)\}_n$ be a sequence of moderately priced crowdfunding games such that $\lim_{n\rightarrow\infty} \frac{B_n}{n}=q$ for some $q\in[0,1]$ and $\{\sigma^*_n\}$ be the corresponding sequence of non-trivial symmetric equilibria.  Then the limit equilibrium strategy  is:
	\begin{equation}
	\lim_{n\rightarrow\infty} \sigma^*_n(l)=\begin{cases}
	0&\mbox{ if } q\leq 1-p\\
	\frac{q-(1-p)}{p}&\mbox{ otherwise}
	\end{cases}\mbox{ and } \lim_{n\rightarrow\infty} \sigma^*_n(h)=1.
	\end{equation}
\end{customlemma}

\begin{proof}
	By Lemma \ref{lem:dominant_one}, $\sigma^*(h)=1$ in any crowdfunding game with moderate price. This also hold in the limit.
	
	Next we show that if $q<1-p$ then $\lim_{n\rightarrow\infty} \sigma^*_n(l)=0.$ Assume by contradiction that $\lim_{n\rightarrow\infty} \sigma^*_n(l)=\lambda>0.$
	Let $c^L_n=E_{\Gamma_n,\sigma^*_n}\sum_i a_i|\omega=L.$ That is $C_n$  denote the expected number of contributors in state $L.$ Recall that by the law of large numbers $$\lim_{n\rightarrow\infty}\frac{c^L_n}{n}=\alpha^L=p+(1-p)\lambda>1-p>q,$$ therefore for sufficiently large $n$ the low player expected utility is,
	$$
	E_{\Gamma_n,\sigma^*_n}u(a_i=1|s_i=l)=p_l(1-\tau)-(1-p_l)\tau<0
	$$
	Where 	$$p_l=Pr(\omega=H|s_i=l)=\frac{\mu(1-p)}{p(1-\mu)+\mu(1-p)}.$$ The inequality holds as prices are moderate.
	
	Finally we show that if $q\ge 1-p$ then  $\lim_{n\rightarrow\infty} \sigma^*_n(l)=\frac{q-(1-p)}{p}.$ For any $\lambda,$ if $\lim_{n\rightarrow\infty} \sigma^*_n(l)=\lambda$ then by the law of large numbers,
	$$
	\lim_{n\rightarrow\infty}\frac{c^L_n}{n}=\alpha^L_\infty=p+(1-p)\lambda.	
	$$
	If $\lambda<\frac{q-(1-p)}{p},$ then $\lim_{n\rightarrow\infty}\varphi(\alpha^L_n,n-1,B-1)=0,$ and thus,
	$$\lim_{n\rightarrow\infty}E_{\Gamma_n,\sigma^*}u(a=1|s=l)=p_l(1-\tau)>0.$$ Similarly if $\lambda>\frac{q-(1-p)}{p},$ then $\lim_{n\rightarrow\infty}\varphi(\alpha^L_n,n-1,B-1)=1$ and thus
	$$\lim_{n\rightarrow\infty}E_{\Gamma_n,\sigma_n^*}u(a=1|s=l)=p_l(1-\tau)-(1-p_l)\tau<0,$$ where again, the inequality holds as prices are moderate.
\end{proof}

The following corollary is immediate by Lemma \ref{lem:asympt_probs}.
\begin{corollary}\label{cor:mod_price_H_1}
	For any $\mu,p,$ a moderate price $\tau$ and a sequence $\{B_n\}$ s.t. $\limninf \frac{B_n}{n}=q$ for some $q\in[0,1],$ $\limninf \varphi(\alpha^H_n,n-1,B_n-1)=1.$
\end{corollary}
In words, Corollary \ref{cor:mod_price_H_1} states that if the price is moderate, the probability that a campaign succeeds when the state is $H$ approaches one as population grows.

\begin{proof}
	By  the Law of Large numbers, the expected number of contribution conditional on the state being $H$ is,
	$
	\limninf \frac{c^H_n}{n}=\limninf \alpha^H_n
	$ where $\alpha^H_n$ is the probability of a player choosing action $1$ in the non-trivial, symmetric equilibrium of game $\Gamma_n.$ By Lemma \ref{lem:asympt_probs}
	$$
	\limninf \alpha^H_n=p+(1-p)\frac{q-(1-p)}{p}>q.
	$$
\end{proof}

\begin{customthm}{\ref{thm:asym_correctness}}
	For any large crowdfunding game with prior $\mu$, signal quality $p$ and a moderate price $\tau$ the probability of making the correct choice is given by:
	\begin{equation}\label{eq:max}
	\theta(\mu,p,\tau) =1-\frac{1-p}{p}\frac{1-\tau}{\tau}\mu .
	\end{equation}
\end{customthm}

\begin{proof}
	First we show that $\theta(\mu,p,\tau) \ge1-\frac{1-p}{p}\frac{1-\tau}{\tau}\mu.$
	Consider a sequence of thresholds $\{B_n\}$ such that $B_n=\frac{n}{2}$ and the corresponding sequence of moderately priced crowdfunding games $\{\Gamma_n\}$ where $\Gamma_n=\Gamma(n,B_n,\mu,p,\tau).$
	As $\limninf B_n=\frac{1}{2}>1-p,$ by Lemma \ref{lem:asympt_probs} we get that eventually,  $\sigma^*_n(l)\in(0,1)$ and thus in equilibrium, the low player is indifferent between both actions yielding,
	$$
	\lim_{n\rightarrow\infty}E_{\Gamma_n,\sigma^*_n}u(a=1|s=l)=0\Leftrightarrow p_l\varphi(\alpha_\infty^H,n-1,B-1)(1-\tau)-(1-p_l)\varphi(\alpha^L_\infty,n-1,B-1)\tau=0
	$$
	where $p_l=Pr(\omega=H|s_i=l)=\frac{\mu(1-p)}{\mu(1-p)+p(1-\mu)}.$
	Recall that $c^\omega_n=E_{\Gamma_n,\sigma^*_n}\sum_i(a_i)|\omega$ is the expected number of contributions in state $\omega\in H,L.$ By the law of large numbers $$\lim_{n\rightarrow\infty}\frac{c^\omega_n}{n}=\alpha^\omega_n,$$
	where, as before, $\alpha_n^\omega$ is the probability that a player plays $1$ in the equilibrium of $\Gamma^n$ when the state is $\omega\in H,L.$
	As $\sigma^*(h)=1$ we get that $\alpha^H>p>\frac{1}{2}$ and thus, when the state is $H,$ the probability that the product be realized approaches $1$ that is $$\lim_{n\rightarrow\infty}\varphi(\alpha^H_n,n-1,B-1)=1.$$ By the low player indifference condition, we can therefore calculate the probability the threshold be reached at the limit.
	\begin{eqnarray*}
	0=\lim_{n\rightarrow\infty}E_{\Gamma_n,\sigma^*_n}u(a=1|s=l)=\frac{(1-p)\mu}{(1-p)\mu+(1-\mu)p}(1-\tau)& &\\
	-\frac{(1-\mu)p}{(1-p)\mu+(1-\mu)p}\tau\lim_{n\rightarrow\infty}\varphi(\alpha_n^L,n-1,B-1).& &
	\end{eqnarray*}
	Rearranging the above equality we get,
	$$
	\lim_{n\rightarrow\infty}\varphi(\alpha_n^L,n-1,B-1)=\frac{1-p}{p}\frac{\mu}{1-\mu}\frac{1-\tau}{\tau}.
	$$
	By definition of the correctness index,
	$$
	\lim_{n\rightarrow\infty}\theta(n,B_n=\frac{n}{2},\mu,p,\tau)=\mu+(1-\mu)\frac{1-p}{p}\frac{\mu}{1-\mu}\frac{1-\tau}{\tau}=1-\frac{1-p}{p}\frac{1-\tau}{\tau}\mu.
	$$
	Therefore the expression in Theorem \ref{lem:asympt_probs} is achievable.
	
	To complete the proof we must show that $\theta(\mu,p,\tau)\le 1-\frac{1-p}{p}\frac{1-\tau}{\tau}\mu.$ Let $\{B^*_n\}$ be a sequence of thresholds for which the maximal correctness is achieved.
	By taking a subsequence if necessary, the limit
	$\lim_{n\rightarrow\infty}B^*_n=q^*$
	exists for some $q^*\in[0,1].$
	
	Next we show that $q^*>1-p.$ Assume to the contrary that $q^*<1-p.$ By Theorem \ref{thm:unique_eq} for any $n,$ $\sigma^*_n(h)=1$ therefore by the Law of large numbers, the expected number of contributions when the state is $L$ is $\limninf\frac{c^L_n}{n}=\alpha^L_n>1-p>q^*$ and thus the probability that the threshold is crossed when the state is $L$ is,
	$$\limninf\varphi(\alpha^L_n,n-1,B-1)=1$$ and hence, by definition of the correctness index and under the contrary assumption,
	$$
	\theta(p,\mu,\tau)=\mu.
	$$
	A contradiction as prices are moderate and thus,
	$$
	1-\frac{1-p}{p}\frac{1-\tau}{\tau}\mu>\mu\Leftrightarrow\frac{1-\mu}{\mu}>\frac{1-p}{p}\frac{1-\tau}{\tau}\Leftrightarrow \frac{(1-p)\mu}{(1-\mu)p+(1-p)\mu}>\tau,
	$$ and we have seen that the expression on the right is achievable. Therefore it must be that the optimal sequence of thresholds is such that,
	$$
	\limninf \frac{B^*_n}{n}=q^*>1-p.
	$$
	
	As $q^*>1-p,$ by Lemma \ref{lem:asympt_probs}, we get that the low player is eventually indifferent in $\Gamma_n$ and thus for sufficiently large $n,$
	\begin{equation}\label{eq:pre-linear-prog}
	\begin{split}
	&(1-p)\mu(1-\tau)\varphi(\alpha_n^H,n-1,B_n-1)-p(1-\mu)\tau\varphi(\alpha_n^L,n-1,B_n-1)\Rightarrow\\
	&\varphi(\alpha^H_n,n-1,B_n-1)=\frac{p(1-\mu)\tau}{(1-p)\mu(1-\tau)}\varphi(\alpha_n^L,n-1,B_n-1).
	\end{split}
	\end{equation}
	
	Assume to the contrary that for some $\mu,p$ and a moderate price $\tau,$
	$$\theta(\mu,p,\tau)>1-\frac{1-p}{p}\frac{1-\tau}{\tau}\mu=\mu+(1-\mu)(1-\frac{1-p}{p}\frac{\mu}{1-\mu}\frac{1-\tau}{\tau}).$$ by the definition of correctness and the inequality above, This entails that there exists a sequence of thresholds $\{B^*_n\}$ and a sequence of corresponding games $\Gamma_n$ such that $\limninf\varphi(\alpha^L_n,n-1,B_n-1)<\frac{1-p}{p}\frac{\mu}{1-\mu}\frac{1-\tau}{\tau}.$
	Therefore equation \eqref{eq:pre-linear-prog} becomes,
	\begin{equation}\label{eq:pre-linear-prog1}
	\begin{split}
	&\varphi(\alpha_n^H,n-1,B_n-1)=\frac{p(1-\mu)\tau}{(1-p)\mu(1-\tau)}\varphi(\alpha^L_n,n-1,B_n-1)<1,
	\end{split}
	\end{equation}
	a contradiction as by Lemma \ref{lem:asympt_probs}, the expected proportion of contributions when the state is $H$ approaches
	$$\limninf \frac{c^H_n}{n}=\limninf\alpha^H_n=p+(1-p)\frac{q^*-(1-p)}{p}>q^*.$$
	and thus, by the law of large numbers, $\limninf\varphi(\alpha_n^H,n-1,B_n-1)=1.$
\end{proof}

\begin{customthm}{\ref{thm:revenue}} For any crowdfunding game with prior $\mu$, signal quality $p$ and a moderate price $\tau$, the participation ratio in a large campaign is given by:
	\begin{equation}\label{eq:revenue}
	R(\mu,p,\tau) =\mu(1+\frac{1-p}{p}\frac{1-\tau}{\tau})
	\end{equation}
\end{customthm}
\begin{proof}
	
	First we show that for any $\mu,p$ and a moderate price $\tau,$ $R(\mu,p,\tau)\ge \mu(1+\frac{1-p}{p}\frac{1-\tau}{\tau}).$ Let $\{B_n\}$ be a sequence of thresholds such that $\limninf\frac{B_n}{n}=q$ for some $q>1-p.$ By Lemma \ref{lem:asympt_probs}, for sufficently large $n,$ the low player is indifferent between both actions and thus,
	\begin{equation*}
	\begin{split}
	&\limninf E_{\Gamma^n,\sigma^*_n} u(a=1|s=l)=0\Leftrightarrow\\
	&\mu(1-p)(1-\tau)\limninf\varphi(\alpha_n^H,n-1,B_n-1)=(1-\mu)p\tau \limninf\varphi(\alpha^L_n,n-1,B_n-1).
	\end{split}
	\end{equation*}
	By Corollary \ref{cor:mod_price_H_1} $\limninf\varphi(\alpha^H_n,n-1,B_n-1)=1$ and thus,
	\begin{equation}\label{eq:prob_l1}
	\limninf\varphi(\alpha^L_n,n-1,B_n-1)=\frac{\mu(1-p)(1-\tau)}{(1-\mu)p\tau}.
	\end{equation}
	
	Recall that by Lemma \ref{lem:almost_unpivotal}, in large crowdfunding games players are (almost) non-pivotal as,
	\begin{equation}
	\begin{split}
	&\limninf\varphi(\alpha^\omega_n,n-1,B_n-1)=\limninf Pr_{\Gamma_n,\sigma^*_n} (c^\omega_n-1\ge B_n-1)=\\
	&\limninf Pr_{\Gamma_n,\sigma^*_n} (\frac{c^\omega_n}{n}-\frac{1}{n}>q-\frac{1}{n})=
	\limninf\varphi(\alpha^\omega_n,n,B_n).
	\end{split}
	\end{equation}
	Thus, by the definition of Participation
	\begin{equation}
	\begin{split}
	&\limninf R(B_n,n\mu,p,\tau)\ge q(\mu\limninf\varphi(\alpha^H_n,n,B_n)+(1-\mu)\limninf\varphi(\alpha^L_n,n,B_n))=\\
	&q(\mu+(1-\mu)\frac{\mu(1-p)(1-\tau)}{(1-\mu)p\tau}).
	\end{split}
	\end{equation}
	By taking $q$ to $1$ we can see that,
	$$
	R(\mu,p,\tau)\ge (\mu+(1-\mu)\frac{\mu(1-p)(1-\tau)}{(1-\mu)p\tau})=\mu(1+\frac{1-p}{p}\frac{1-\tau}{\tau}).
	$$
	
	Next we prove that $R(\mu,p,\tau)\le\mu(1+\frac{1-p}{p}\frac{1-\tau}{\tau}).$
	By Corollary \ref{cor:mod_price_H_1}, $\limninf \varphi(\alpha^H_n,n-1,B_n-1)=1.$ Next note that
	\begin{flalign*}
	&\limninf\varphi(\alpha^H_n,n-1,B_n-1)= \\
	&\limninf Pr_{\Gamma_n,\sigma^*_n} (\frac{c^H_n-1}{n}\ge \frac{B_n-1}{n})=\\
	&\limninf Pr_{\Gamma_n,\sigma^*_n} (\frac{c^H_n}{n}\ge \frac{B_n}{n})=\varphi(\alpha^H_n,n,B_n). 
	\end{flalign*}
	Thus the maximal participation is bounded from above by
	$$
	R(\mu,p,\tau)\le \mu+(1-\mu)y^*
	$$
	where $y^*$ is the solution of the following linear programming:
	\begin{equation}
	\begin{aligned}
	& {\text{max}}
	& & \mu + (1-\mu )y \\
	& \text{s.t.} & & 1\geq  y\geq 0 \\
	& & &  (1-p)\mu (1-\tau)-p(1-\mu)y\tau\ge 0. \\
	\end{aligned}
	\end{equation}
	Which is simply $y^*=\frac{(1-p)\mu(1-\tau)}{(1-p)\mu\tau}$ and thus
	$$
	R(\mu,p,\tau)\le\mu(1+\frac{1-p}{p}\frac{1-\tau}{\tau}).	
	$$
\end{proof}

\subsection{Proofs for expensive prices.}\label{sec:expensive_price_proofs}

\begin{customthm}{\ref{lem:complement_11}}
	Let $\{\Gamma(n,B_n,\mu,p,\tau)\}_n$ be a sequence of expensively priced crowdfunding games such that $\lim_{n\rightarrow\infty} \frac{B_n}{n}=q$ for some $q\in[0,1]$. Then the limit equilibrium strategy is:
	\begin{equation}
	\lim_{n\rightarrow\infty} \sigma^*_n(l)=\begin{cases}
	0&\mbox{ if } q\leq 1-p\\
	\frac{q-(1-p)}{p}&\mbox{ otherwise}
	\end{cases}\mbox{ and }
	\lim_{n\rightarrow\infty} \sigma^*_n(h)= 	
	\begin{cases}
	\frac{q}{1-p}&\mbox{ if } q\leq 1-p\\
	1&\mbox{ otherwise}
	\end{cases}
	\end{equation}
\end{customthm}

%
%\begin{lemma}\label{lem:asympt_probs_e_price}
%	Let $\{\Gamma(n,B_n,\mu,p,\tau)\}_n$ be a sequence of expensively priced crowdfunding games such that $\lim_{n\rightarrow\infty} \frac{B_n}{n}=q$ for some $q\in[0,1]$. Then the limit equilibrium strategy is:
%	\begin{equation}
%	\lim_{n\rightarrow\infty} \sigma^*_n(l)=\begin{cases}
%	0&\mbox{ if } q\leq 1-p\\
%	\frac{q-(1-p)}{p}&\mbox{ otherwise}
%	\end{cases}\mbox{ and }
%	\lim_{n\rightarrow\infty} \sigma^*_n(h)= 	
%	\begin{cases}
%	\frac{q}{1-p}&\mbox{ if } q\leq 1-p\\
%	1&\mbox{ otherwise}
%	\end{cases}
%	\end{equation}
%\end{lemma}

\begin{proof}
	
	First we prove the lemma for all sequences for which $q>1-p.$ We start by showing that in this case $\lim_{n\rightarrow\infty} \sigma^*_n(h)=1.$ Assume to the contrary that there exists some $\lim_{n\rightarrow\infty} \frac{B_n}{n}=q>1-p$ for which $\lim_{n\rightarrow\infty} \sigma^*_n(h)<1.$
	Then, by Lemma \ref{lem:psi_ge_lambda}, for sufficiently large $n$ it must be that $\sigma^*_n(l)=0.$  Therefore, by the Law of Large Numbers,
	$$\limninf c^L_n=\limninf\alpha^L_n<1-p<q$$ and thus $\limninf\varphi(\alpha^L_n,n-1,B_n-1)=0.$
	In addition, as $\sigma^*_n(h)<1$ the following inequality must be satisfied for the high-type player,
	$$
	\limninf (p\mu(1-\tau)\varphi(\alpha^H_n,n-1,B_n-1)-(1-p)(1-\mu)\tau\varphi(\alpha^L_n,n-1,B_n-1))\le 0.
	$$
	This entails a contradiction as $\sigma^*_n$ is a non trivial equilibrium and thus $\varphi(\alpha^H_n,B_n-1,n-1)>0.$
	
	As $\limninf \sigma^*_n(h)=1,$ the proof that $\limninf \sigma^*_n(l)=\frac{q-(1-p)}{p}$ is similar to the case of moderate prices and thus omitted.
	
	Next we prove the lemma for all sequences for which $q\le 1-p.$  We start by showing that in this case $\lim_{n\rightarrow\infty} \sigma^*_n(l)=0.$ Assume to the contrary that there exists some $\lim_{n\rightarrow\infty} \frac{B_n}{n}=q\le 1-p$ for which $\lim_{n\rightarrow\infty} \sigma^*_n(l)>0.$ Then, by Lemma \ref{lem:psi_ge_lambda}, for sufficiently large $n$ it must be that $\sigma^*_n(h)=1.$  Therefore, by the Law of Large Numbers,
	$$\limninf c^L_n= \limninf\alpha^L_n>1-p\ge q$$ and thus $\limninf\varphi(\alpha^L_n,n-1,B_n-1)=1.$
	In addition, as $\sigma^*_n(l)>0$ the following inequality must be satisfied for the low-type player,
	\begin{flalign*}
	&\limninf ((1-p)\mu(1-\tau)\varphi(\alpha^H_n,n-1,B_n-1)-p (1-\mu)\tau\varphi(\alpha^L_n,n-1,B_n-1))=\\
	&(1-p)\mu(1-\tau)-p (1-\mu)\tau\ge 0.
	\end{flalign*}
	This entails a contradiction as the price is expensive. Therefore $\lim_{n\rightarrow\infty} \sigma^*_n(l)=0.$
	
	Next assume by contradiction that there exist a sequence  $\lim_{n\rightarrow\infty} \frac{B_n}{n}=q\le 1-p$ for which  $\lim_{n\rightarrow\infty} \sigma^*_n(h)>\frac{q}{1-p}.$ By the law of large numbers, $$\limninf c^L_n= \limninf\alpha^L_n=(1-p)\limninf\sigma^*_n(h)\ge q$$ and thus $\limninf\varphi(\alpha^L_n,n-1,B_n-1)=1.$
	
	As $\lim_{n\rightarrow\infty} \sigma^*_n(h)>0$ the following inequality must be satisfied for high players,
	\begin{flalign*}
	&\limninf p\mu(1-\tau) \varphi(\alpha^H_n,n-1,B_n-1)\\
	&-(1-p)(1-\mu)\tau \varphi(\alpha^L_n,n-1,B_n-1)\ge 0\\
	&\Leftrightarrow
	  p\mu(1-\tau) - (1-p)(1-\mu)\tau\ge 0,	
	\end{flalign*}
	a contradiction as the price is expensive.
	
	Finally, assume by contradiction that there exist a sequence  $\lim_{n\rightarrow\infty} \frac{B_n}{n}=q\le 1-p$ for which  $\lim_{n\rightarrow\infty} \sigma^*_n(h)<\frac{q}{1-p}.$ By the law of large numbers, $$\limninf c^L_n= \limninf\alpha^L_n=(1-p)\limninf\sigma^*_n(h)< q$$ and thus $\limninf\varphi(\alpha^L_n,n-1,B_n-1)=0.$
	
	As $\lim_{n\rightarrow\infty} \sigma^*_n(l)=0$ the following inequality must be satisfied for low players,
	\begin{flalign*}
	&\limninf p\mu(1-\tau) \varphi(\alpha^H_n,n-1,B_n-1)-\\
	&(1-p)(1-\mu)\tau \varphi(\alpha^L_n,n-1,B_n-1)\le 0\Leftrightarrow\\
	&  p\mu(1-\tau)\limninf\varphi(\alpha^H_n,,n-1,B_n-1)\le 0,	
	\end{flalign*}
	a contradiction as the equilibrium is non-trivial.
\end{proof}

\begin{customthm}{\ref{thm:correct_asym_game_e_price}}
	\begin{equation}
	\theta(\mu,p,\tau)=
	1-\frac{1-p}{p}\frac{1-\tau}{\tau}\mu.
	\end{equation}
\end{customthm}
\begin{proof}
	First we show that 	$\theta(\mu,p,\tau)\ge	1-\frac{1-p}{p}\frac{1-\tau}{\tau}\mu.$ Consider the sequence $\{B_n\}$ where $B_n=\frac{n}{2}$ and thus $\limninf \frac{B_n}{n}=\frac{1}{2}>1-p.$ By Theorem \ref{lem:complement_11} the limit of the sequence of non-trivial Bayes-Nash equilibria is
	\begin{equation}\label{eq:e_price_lim_probs}
	\limninf \sigma^*_n(s_i)=\begin{cases}
	1&\mbox{ if }s_i=h\\
	\frac{\frac{1}{2}-(1-p)}{p}=\frac{2p-1}{p}&\mbox{ if }s_i=l
	\end{cases}.
	\end{equation}
	Therefore for sufficiently large $n$, the low player must be indifferent between both actions and the following equality must be satisfied,
	\begin{flalign*}
	&E_{\sigma^*_n} u(a=1|s=l)=0\Leftrightarrow\\
	& (1-p)\mu(1-\tau)\limninf\varphi(\alpha^H_n,n-1,B_n-1)\\
	&-p(1-\mu)\tau\limninf\varphi(\alpha^L_n,n-1,B_n-1)=0.
	\end{flalign*}
	
	In addition, by equation \ref{eq:e_price_lim_probs}, $\limninf c^H_n=\limninf\alpha^H_n=p+(1-p)\frac{2p-1}{p}>p>\frac{1}{2}$ and thus $\limninf\varphi(\alpha^H_n,n-1,B_n-1)=1$ which yields,
	$$
	\limninf\varphi(\alpha^L_n,n-1,B_n-1)=\frac{1-p}{p}\frac{\mu}{1-\mu}\frac{1-\tau}{\tau}.	
	$$
	
	Recall that by Lemma \ref{lem:almost_unpivotal}, in large crowdfunding games players are (almost) non-pivotal as,
	\begin{equation}
	\begin{split}
	&\limninf\varphi(\alpha^\omega_n,n-1,B_n-1)=\limninf Pr_{\Gamma_n,\sigma^*_n} (c^\omega_n-1\ge B_n-1)=\\
	&\limninf Pr_{\Gamma_n,\sigma^*_n} (\frac{c^\omega_n}{n}-\frac{1}{n}>q-\frac{1}{n})=
	\limninf\varphi(\alpha^\omega_n,n,B_n).
	\end{split}
	\end{equation} Thus, by definition of the correctness index,
	$$
	\theta(\mu,p,\tau)\ge\limninf \theta(n,\frac{n}{2},\mu,p,\tau)=\mu+(1-\mu)(1-\frac{1-p}{p}\frac{\mu}{1-\mu}\frac{1-\tau}{\tau})=1-\frac{1-p}{p}\frac{1-\tau}{\tau}\mu.
	$$
	
	Next we show that $\theta(\mu,p,\tau)\le 1-\frac{1-p}{p}\frac{1-\tau}{\tau}\mu.$ Let $\{B^*_n\}$ be a sequence of thresholds for which the optimal correction is achieved. We show that $\limninf \frac{B^*_n}{n}=q^*>1-p.$ Assume by contradiction that $\limninf \frac{B^*_n}{n}=q^*\le 1-p.$ By Lemma \ref{eq:e_price_lim_probs}, this entails that,
	$$
	\limninf c^H_n=\limninf\alpha^H_n=p\frac{q^*}{1-p}>q^*,
	$$
	and thus $\limninf \varphi(\limninf\alpha^H_\infty,B^*_n-1,n-1)=1.$
	Again by Lemma \ref{eq:e_price_lim_probs}, the high player is eventually indifferent between both actions and thus the following equality must be satisfied,
	$$
	\limninf E_{\sigma^*_n} u(a=1|s=h)=0\Leftrightarrow p\mu(1-\tau)-(1-p)(1-\mu)\tau \limninf \varphi(\limninf\alpha^L_n,n-1,B^*_n-1)=0
	$$
	and thus $\varphi(\limninf\alpha^L_n,B^*_n-1,n-1)=\frac{p\mu(1-\tau)}{(1-p)(1-\mu)\tau}.$
	
	By the definition of correctness and by Lemma \ref{lem:almost_unpivotal},
	$$
	\limninf \theta(n,B^*_n,\mu,p,\tau)=\mu+(1-\mu)(1-\frac{p\mu(1-\tau)}{(1-p)(1-\mu)\tau})=1-\frac{p}{1-p}\frac{1-\tau}{\tau}\mu<1-\frac{1-p}{p}\frac{1-\tau}{\tau}\mu,
	$$ a contradiction as we saw that the latter expression can be achieved for $B_n=\frac{n}{2}.$
	
	As $\limninf \frac{B^*_n}{n}=q^*>1-p$, by Lemma \ref{eq:e_price_lim_probs}, high player eventually surely contribute and low players eventually become indifferent between choosing both actions, thus  $\limninf \varphi(\alpha^H_nn-1,B^*_n-1)=1.$ Therefore the maximal correctness index is bounded from above by the following linear programming problem,
	\begin{equation*}
	\begin{aligned}
	& {\text{max}}
	& & \mu +(1-\mu)(1-y) \\
	& \text{s.t.} & & 1\geq y\geq 0 \\
	& & &  (1-p)\mu (1-\tau)-p(1-\mu)y\tau=0. \\
	\end{aligned}
	\end{equation*}
	Which is simply
	$$y=\frac{1-p}{p}{\mu}{1-\mu}\frac{1-\tau}{\tau}.$$ and thus by definition of asymptotic correctness,
	$$
	\theta(\mu,p,\tau)\le 1-\frac{1-p}{p}\frac{1-\tau}{\tau}\mu.
	$$
\end{proof}

\begin{customthm}{\ref{thm:part_asym_game1}}
	\begin{itemize}
		\item  If $\mu<\frac{1}{3}$ and $p\le\sqrt{3}-1,$ or if $\mu<\frac{1}{3}, p>\sqrt{3}-1$ and $\tau>\frac{2\mu }{(1-\mu)p+2(1-p)\mu)}$ then	 	
		\begin{equation*}
		\lim_{n\rightarrow\infty}\max_{B\in\{1\dots n\}} R_\mu(n,B,\tau)=\mu p+(1-\mu)\frac{1-p}{2}.
		\end{equation*}
		\item Otherwise,	
		\begin{equation*}
		\lim_{n\rightarrow\infty}\max_{B\in\{1\dots n\}} R_\mu(n,B,\tau)=
		\mu(1+\frac{1-p}{p}\frac{1-\tau}{\tau})
		\end{equation*}	
	\end{itemize}
\end{customthm}
\begin{proof}
	First note that in any game $\Gamma(\Params),$ the participation is bounded by
%	\begingroup\makeatletter\def\f@size{6.5}\check@mathfonts
	\begin{flalign}\label{eq:bounds_of_R}
	&1(\mu\varphi(\alpha^H,n-1,B-1)+(1-\mu) \varphi(\alpha^L,n-1,B-1))\ge R(\Params)\ge\\
	& \ge\frac{B}{n}(\mu\varphi(\alpha^H,n-1,B-1)+(1-\mu) \varphi(\alpha^L,n-1,B-1)).\nonumber
	\end{flalign}
%	\endgroup
	
	Next we show that the asymptotic participation  is bounded from bellow by the following expression
	$$R(\mu,p,\tau)\ge \mu(1+\frac{1-p}{p}\frac{1-\tau}{\tau}).$$
	Consider a sequence $\{B_n\}$ such that $B_n=\frac{n}{2}.$ As $\limninf \frac{B_n}{n}=\frac{1}{2}>1-p,$ by Theorem \ref{lem:complement_11} and the law of large numbers we get that the proportion of contribution when the state is $H$ converges to,
	$\limninf c^H_n=\alpha^H_n=p+(1-p)\frac{\frac{1}{2}-(1-p)}{p}>\frac{1}{2}$ and thus the probability that the product succeeds in state $H$ approaches $\limninf\varphi(\alpha^H_n,n,B_n)=1.$
	
	In addition, by Lemma \ref{lem:asympt_probs_e_price}, the low player is eventually indifferent between both actions, thus the following equality is satisfied for large $n$s,
	$$
	\limninf E_{\sigma^*_n}u(a=1|s=l)=0\Rightarrow (1-p)\mu(1-\tau)-p(1-\mu)\tau\limninf\varphi(\alpha^H_n,B_n-1,n-1)=0,
	$$
	and thus $\limninf\varphi(\alpha^H_n,B_n-1,n-1)=\frac{1-p}{p}\frac{1-\tau}{\tau}\frac{\mu}{1-\mu}.$
	Recall that, by Lemma \ref{lem:almost_unpivotal} in large crowdfunding games players are (almost) non-pivotal as,
	\begin{equation}
	\begin{split}
	&\limninf\varphi(\alpha^\omega_n,n-1,B_n-1)=\limninf Pr_{\Gamma_n,\sigma^*_n} (c^\omega_n-1\ge B_n-1)=\\
	&\limninf Pr_{\Gamma_n,\sigma^*_n} (\frac{c^\omega_n}{n}-\frac{1}{n}>q-\frac{1}{n})=
	\limninf\varphi(\alpha^\omega_n,n,B_n).
	\end{split}
	\end{equation}
	
	By the definition of participation and asymptotic participation then,
	\begin{equation}\label{eq:con_1_e_price_R}
	R(\mu,p,\tau)\ge\limninf R(n,\frac{n}{2},\mu,p,\tau)=\mu(1+\frac{1-p}{p}\frac{1-\tau}{\tau}).
	\end{equation}
	
	Next let $\{B^*_n\}$ be the sequence for which the maximal asymptotic participation is achieved. We divide the proof into two cases. (1) where $\limninf\frac{B^*_n}{n}=q^*\ge 1-p$ and (2)$\limninf\frac{B^*_n}{n}=q^*< 1-p.$
	
	\textbf{Case (1) $\limninf\frac{B^*_n}{n}=q^*\ge 1-p$:} Similarily to our discussion for the sequence  $\{\frac{n}{2}$ above, by Lemma \ref{eq:e_price_lim_probs}, eventually high players contribute and low players are indifferent between both actions and thus $\limninf\varphi(\alpha^H_n,n-1,B^*_n-1)=1$ and $\limninf\varphi(\alpha^L_n,n-1,B^*_n-1)=\frac{1-p}{p}\frac{1-\tau}{\tau}\frac{\mu}{1-\mu}.$
	Therefore by equation \ref{eq:bounds_of_R}, the following condition is satisfied,
	\begin{equation}
	\limninf R(n,B^*n,\mu,p,\tau)\le \mu+(1-\mu)\frac{1-p}{p}\frac{1-\tau}{\tau}\frac{\mu}{1-\mu}=\mu(1+\frac{1-p}{p}\frac{1-\tau}{\tau}).
	\end{equation}
	This, together with equation \eqref{eq:con_1_e_price_R} yields that if the maximal asymptotic participation is acheived for a sequence where  $\limninf\frac{B^*_n}{n}=q^*\ge 1-p,$ then
	$$
	R(\mu,p,\tau=\mu(1+\frac{1-p}{p}\frac{1-\tau}{\tau}))\equiv R_1.
	$$
	We denote this expression by $R_1.$
	
	\textbf{Case (2) $\limninf\frac{B^*_n}{n}=q^*< 1-p$:} 	We can calculate the expected number of contributions provided that the threshold is reached using the laws of conditional expectation. First we present the calculations for $\omega=H.$ As before, one can see that in state $H$ the campaign will eventually succeeds. To see that recall by Theorem \ref{lem:complement_11},
	$$
	\limninf c^H_n=\alpha_H(\limninf\sigma^*_n)=p\frac{q}{1-p}>q
	$$
	and hence by the law of large numbers
	$\limninf\varphi(\alpha^H_n,n-1,B_n-1)=1.$ Therefore, at $H,$ the sum of contributions conditional on the campaign success equals the expected number of contributions and,
	\begin{equation}\label{eq:e_price_con1}
	\limninf c^H_n|_{c^H_n\ge B_n-1}=\limninf \frac{c^H_n}{n}=p\frac{q}{1-p}.
	\end{equation}
	Next we calculate the expected number of contributions when $omega=L$ and the population size increases.  We will use the "Binomial approximation to Normal distribution" (see \cite{Feller1968}).
	
	By \cite{Feller1968} pp. 174 - 187,
	Let $Z\sim Bin(n,\gamma)$ be a random variable for some $\gamma<1$ and define the transformation
	$$
	R=\frac{Z-n\gamma}{\sqrt{n\gamma(1-\gamma)}} .
	$$
	By DeMoivre-Laplace limit theorem, as $q<1-p<0.5$ and $n\rightarrow \infty$, the distribution of $R$ approaches  the standard normal distribution (See \cite{Feller1968} page 186, Theorem 2 and equation 3.18.). Therefore
	$$\lim_{n\rightarrow\infty}Pr(Z>n\gamma)=\lim_{n\rightarrow\infty}Pr(R>0)=\frac{1}{2}.$$
	Next we calculate the conditional expected value of $Z,$
	\begin{equation}\label{eq:Bin_appx}
	\begin{split}
	&E(Z|Z>n\gamma)=E(\sqrt{n\gamma(1-\gamma)}*R+n\gamma|\sqrt{n\gamma(1-\gamma)}*R+n\gamma>n\gamma)=\\
	&\sqrt{n\gamma(1-\gamma)}E(R|\sqrt{n\gamma(1-\gamma)}*R>0)+n\gamma=\\
	&\sqrt{n\gamma(1-\gamma)}E(R|R>0)+n\gamma=\\
	&\sqrt{n\gamma(1-\gamma)}\sqrt{\frac{2}{\pi}}+n\gamma.
	\end{split}
	\end{equation}
	
	As $q<1-p,$ By Theorem \ref{lem:complement_11} we get that
	$$
	\limninf c^L_n = \alpha^L_n=(1-p)\frac{q}{1-p}=q.
	$$
	
	Therefore, by equation \eqref{eq:Bin_appx},
	\begin{equation}\label{eq:e_price_con2}
	\limninf E_{\sigma^*_n} c^L_n|\frac{c^L_n}{n}>q = \limninf\frac{1}{2}(\sqrt{\frac{2}{\pi}\frac{q(1-q)}{n}}+q)=\frac{q}{2}.
	\end{equation}
	
	We plug in equations \eqref{eq:e_price_con1} and \eqref{eq:e_price_con2} and by the law of large numbers,
	\begin{equation}\label{eq:e_price_con2+1}
	\limninf R(n,B_n,\mu,p,\tau)=\mu p\frac{q}{1-p}+(1-\mu)\frac{q}{2}.
	\end{equation}
	The expression in equation \eqref{eq:e_price_con2+1} increases in $q$ therefore it will be reached when $q^*\rightarrow (1-p)^-.$
	Therefore if $\limninf \frac{B^*_n}{n}=q^*<1-p$ we get that
	$$
	R(\mu,p,\tau)=p\mu+(1-\mu)\frac{1-p}{2}\equiv R_2.
	$$
	We denote the expression by $R_2.$
	
	For any 3-tuple $(\mu,p,\tau)$ of expensive prices, the asymptotic participation index will be $R(\mu,p,\tau)=\max\{R_1,R_2\}.$
	Simple algebraic calculations show that $R_1\ge R_2$ if $\mu\ge\frac{1}{3}$  or if  $\mu<\frac{1}{3}$ and $p\le\sqrt{3}-1$   or if   $\mu<\frac{1}{3}$ and $p>\sqrt{3}-1$ and $\tau>\frac{2\mu}{(1-\mu)p+2(1-p)\mu}.$ Otherwise $R_2>R_1.$
\end{proof}	
\section{Calculations for Crowdfunding Games.}
\label{sec:finite_tab}
\begin{longtable}[c]{|c|c|c|c|c|c|c|c|c|}
	\caption{Calculations of equilibrium strategies and efficiency for Small Crowdfunding Games }
	\label{tab:finite_tab}\\
	\hline
	$\mu$ & $n$     & $p$   & $B$    & $\tau$ & $\psi$ & $\lambda$ & $\Theta$ & $R$   \\ \hline
	\endhead
	0.2 & 40   & 0.55 & 2    & 0.5 & 0 & 0  & 0.800 & 0 \\ \hline
	0.2 & 40   & 0.55 & 20   & 0.5 & 0.894 & 0  & 0.795 & 0.115 \\ \hline
	0.2 & 40   & 0.55 & 2    & 0.7 & 0 & 0  & 0.800 & 0 \\ \hline
	0.2 & 40   & 0.55 & 20   & 0.7 & 0.670 & 0  & 0.807 & 0.009 \\ \hline
	0.2 & 40   & 0.75 & 2    & 0.5 & 0.135 & 0  & 0.670 & 0.040 \\ \hline
	0.2 & 40   & 0.75 & 20   & 0.5 & 1 & 0.156  & 0.953 & 0.182 \\ \hline
	0.2 & 40   & 0.75 & 2    & 0.7 & 0 & 0  & 0.800 & 0 \\ \hline
	0.2 & 40   & 0.75 & 20   & 0.7 & 1 & 0.114  & 0.981 & 0.165 \\ \hline
	0.5 & 40   & 0.55 & 2    & 0.5 & 1 & 0  & 0.500 & 0.500 \\ \hline
	0.5 & 40   & 0.55 & 20   & 0.5 & 1 & 0.164  & 0.606 & 0.517 \\ \hline
	0.5 & 40   & 0.55 & 2    & 0.7 & 0 & 0  & 0.500 & 0 \\ \hline
	0.5 & 40   & 0.55 & 20   & 0.7 & 1 & 0  & 0.736 & 0.313 \\ \hline
	0.5 & 40   & 0.75 & 2    & 0.5 & 1 & 0  & 0.500 & 0.500 \\ \hline
	0.5 & 40   & 0.75 & 20   & 0.5 & 1 & 0.254  & 0.863 & 0.480 \\ \hline
	0.5 & 40   & 0.75 & 2    & 0.7 & 1 & 0  & 0.500 & 0.500 \\ \hline
	0.5 & 40   & 0.75 & 20   & 0.7 & 1 & 0.188  & 0.947 & 0.426 \\ \hline
	0.7 & 40   & 0.55 & 2    & 0.5 & 1 & 1  & 0.700 & 1 \\ \hline
	0.7 & 40   & 0.55 & 20   & 0.5 & 1 & 1  & 0.700 & 1 \\ \hline
	0.7 & 40   & 0.55 & 2    & 0.7 & 1 & 0  & 0.700 & 0.520 \\ \hline
	0.7 & 40   & 0.55 & 20   & 0.7 & 1 & 0.164  & 0.748 & 0.553 \\ \hline
	0.7 & 40   & 0.75 & 2    & 0.5 & 1 & 0  & 0.700 & 0.600 \\ \hline
	0.7 & 40   & 0.75 & 20   & 0.5 & 1 & 0.380  & 0.781 & 0.717 \\ \hline
	0.7 & 40   & 0.75 & 2    & 0.7 & 1 & 0  & 0.700 & 0.600 \\ \hline
	0.7 & 40   & 0.75 & 20   & 0.7 & 1 & 0.254  & 0.918 & 0.613 \\ \hline
	0.9 & 40   & 0.55 & 2    & 0.5 & 1 & 1  & 0.900 & 1 \\ \hline
	0.9 & 40   & 0.55 & 20   & 0.5 & 1 & 1  & 0.900 & 1 \\ \hline
	0.9 & 40   & 0.55 & 2    & 0.7 & 1 & 1  & 0.900 & 1 \\ \hline
	0.9 & 40   & 0.55 & 20   & 0.7 & 1 & 1  & 0.900 & 1 \\ \hline
	0.9 & 40   & 0.75 & 2    & 0.5 & 1 & 1  & 0.900 & 1 \\ \hline
	0.9 & 40   & 0.75 & 20   & 0.5 & 1 & 1  & 0.900 & 1 \\ \hline
	0.9 & 40   & 0.75 & 2    & 0.7 & 1 & 1  & 0.900 & 1 \\ \hline
	0.9 & 40   & 0.75 & 20   & 0.7 & 1 & 1  & 0.900 & 1 \\ \hline
	0.2 & 100  & 0.55 & 5    & 0.5 & 0 & 0  & 0.800 & 0 \\ \hline
	0.2 & 100  & 0.55 & 24   & 0.5 & 0.417 & 0  & 0.796 & 0.047 \\ \hline
	0.2 & 100  & 0.55 & 44   & 0.5 & 0.887 & 0  & 0.786 & 0.171 \\ \hline
	0.2 & 100  & 0.55 & 50   & 0.5 & 1 & 0  & 0.827 & 0.174 \\ \hline
	0.2 & 100  & 0.55 & 98   & 0.5 & 1 & 0.818  & 0.800 & 0.003 \\ \hline
	0.2 & 100  & 0.55 & 5    & 0.7 & 0 & 0  & 0.800 & 0 \\ \hline
	0.2 & 100  & 0.55 & 24   & 0.7 & 0.316 & 0  & 0.807 & 0.004 \\ \hline
	0.2 & 100  & 0.55 & 44   & 0.7 & 0.800 & 0  & 0.859 & 0.073 \\ \hline
	0.2 & 100  & 0.55 & 50   & 0.7 & 0.943 & 0  & 0.875 & 0.106 \\ \hline
	0.2 & 100  & 0.55 & 98   & 0.7 & 1 & 0.748  & 0.800 & 0 \\ \hline
	0.2 & 100  & 0.75 & 5    & 0.5 & 0.204 & 0  & 0.534 & 0.061 \\ \hline
	0.2 & 100  & 0.75 & 24   & 0.5 & 1 & 0  & 0.497 & 0.289 \\ \hline
	0.2 & 100  & 0.75 & 44   & 0.5 & 1 & 0.150  & 0.947 & 0.182 \\ \hline
	0.2 & 100  & 0.75 & 50   & 0.5 & 1 & 0.228  & 0.946 & 0.189 \\ \hline
	0.2 & 100  & 0.75 & 98   & 0.5 & 1 & 0.919  & 0.891 & 0.175 \\ \hline
	0.2 & 100  & 0.75 & 5    & 0.7 & 0.114 & 0  & 0.861 & 0.024 \\ \hline
	0.2 & 100  & 0.75 & 24   & 0.7 & 0.836 & 0  & 0.794 & 0.179 \\ \hline
	0.2 & 100  & 0.75 & 44   & 0.7 & 1 & 0.123  & 0.978 & 0.166 \\ \hline
	0.2 & 100  & 0.75 & 50   & 0.7 & 1 & 0.201  & 0.978 & 0.171 \\ \hline
	0.2 & 100  & 0.75 & 98   & 0.7 & 1 & 0.902  & 0.895 & 0.124 \\ \hline
	0.5 & 100  & 0.55 & 5    & 0.5 & 1 & 0  & 0.500 & 0.500 \\ \hline
	0.5 & 100  & 0.55 & 24   & 0.5 & 1 & 0  & 0.500 & 0.500 \\ \hline
	0.5 & 100  & 0.55 & 44   & 0.5 & 1 & 0.045  & 0.606 & 0.478 \\ \hline
	0.5 & 100  & 0.55 & 50   & 0.5 & 1 & 0.154  & 0.603 & 0.526 \\ \hline
	0.5 & 100  & 0.55 & 98   & 0.5 & 1 & 0.954  & 0.561 & 0.594 \\ \hline
	0.5 & 100  & 0.55 & 5    & 0.7 & 0.021 & 0  & 0.502 & 0 \\ \hline
	0.5 & 100  & 0.55 & 24   & 0.7 & 0.493 & 0  & 0.711 & 0.162 \\ \hline
	0.5 & 100  & 0.55 & 44   & 0.7 & 0.957 & 0  & 0.751 & 0.366 \\ \hline
	0.5 & 100  & 0.55 & 50   & 0.7 & 1 & 0.032  & 0.814 & 0.339 \\ \hline
	0.5 & 100  & 0.55 & 98   & 0.7 & 1 & 0.866  & 0.518 & 0.037 \\ \hline
	0.5 & 100  & 0.75 & 5    & 0.5 & 1 & 0  & 0.500 & 0.500 \\ \hline
	0.5 & 100  & 0.75 & 24   & 0.5 & 1 & 0  & 0.686 & 0.462 \\ \hline
	0.5 & 100  & 0.75 & 44   & 0.5 & 1 & 0.211  & 0.854 & 0.469 \\ \hline
	0.5 & 100  & 0.75 & 50   & 0.5 & 1 & 0.291  & 0.852 & 0.489 \\ \hline
	0.5 & 100  & 0.75 & 98   & 0.5 & 1 & 0.951  & 0.795 & 0.575 \\ \hline
	0.5 & 100  & 0.75 & 5    & 0.7 & 1 & 0  & 0.500 & 0.500 \\ \hline
	0.5 & 100  & 0.75 & 24   & 0.7 & 1 & 0  & 0.686 & 0.462 \\ \hline
	0.5 & 100  & 0.75 & 44   & 0.7 & 1 & 0.170  & 0.941 & 0.423 \\ \hline
	0.5 & 100  & 0.75 & 50   & 0.7 & 1 & 0.249  & 0.940 & 0.437 \\ \hline
	0.5 & 100  & 0.75 & 98   & 0.7 & 1 & 0.931  & 0.823 & 0.422 \\ \hline
	0.7 & 100  & 0.55 & 5    & 0.5 & 1 & 1  & 0.700 & 1 \\ \hline
	0.7 & 100  & 0.55 & 24   & 0.5 & 1 & 1  & 0.700 & 1 \\ \hline
	0.7 & 100  & 0.55 & 44   & 0.5 & 1 & 1  & 0.700 & 1 \\ \hline
	0.7 & 100  & 0.55 & 50   & 0.5 & 1 & 1  & 0.700 & 1 \\ \hline
	0.7 & 100  & 0.55 & 98   & 0.5 & 1 & 1  & 0.700 & 1 \\ \hline
	0.7 & 100  & 0.55 & 5    & 0.7 & 1 & 0  & 0.700 & 0.520 \\ \hline
	0.7 & 100  & 0.55 & 24   & 0.7 & 1 & 0  & 0.700 & 0.520 \\ \hline
	0.7 & 100  & 0.55 & 44   & 0.7 & 1 & 0.045  & 0.762 & 0.514 \\ \hline
	0.7 & 100  & 0.55 & 50   & 0.7 & 1 & 0.154  & 0.760 & 0.562 \\ \hline
	0.7 & 100  & 0.55 & 98   & 0.7 & 1 & 0.954  & 0.602 & 0.618 \\ \hline
	0.7 & 100  & 0.75 & 5    & 0.5 & 1 & 0  & 0.700 & 0.600 \\ \hline
	0.7 & 100  & 0.75 & 24   & 0.5 & 1 & 0.015  & 0.783 & 0.589 \\ \hline
	0.7 & 100  & 0.75 & 44   & 0.5 & 1 & 0.290  & 0.777 & 0.685 \\ \hline
	0.7 & 100  & 0.75 & 50   & 0.5 & 1 & 0.371  & 0.776 & 0.713 \\ \hline
	0.7 & 100  & 0.75 & 98   & 0.5 & 1 & 0.977  & 0.759 & 0.908 \\ \hline
	0.7 & 100  & 0.75 & 5    & 0.7 & 1 & 0  & 0.700 & 0.600 \\ \hline
	0.7 & 100  & 0.75 & 24   & 0.7 & 1 & 0  & 0.811 & 0.577 \\ \hline
	0.7 & 100  & 0.75 & 44   & 0.7 & 1 & 0.211  & 0.912 & 0.603 \\ \hline
	0.7 & 100  & 0.75 & 50   & 0.7 & 1 & 0.291  & 0.911 & 0.623 \\ \hline
	0.7 & 100  & 0.75 & 98   & 0.7 & 1 & 0.951  & 0.827 & 0.692 \\ \hline
	0.9 & 100  & 0.55 & 5    & 0.5 & 1 & 1  & 0.900 & 1 \\ \hline
	0.9 & 100  & 0.55 & 24   & 0.5 & 1 & 1  & 0.900 & 1 \\ \hline
	0.9 & 100  & 0.55 & 44   & 0.5 & 1 & 1  & 0.900 & 1 \\ \hline
	0.9 & 100  & 0.55 & 50   & 0.5 & 1 & 1  & 0.900 & 1 \\ \hline
	0.9 & 100  & 0.55 & 98   & 0.5 & 1 & 1  & 0.900 & 1 \\ \hline
	0.9 & 100  & 0.55 & 5    & 0.7 & 1 & 1  & 0.900 & 1 \\ \hline
	0.9 & 100  & 0.55 & 24   & 0.7 & 1 & 1  & 0.900 & 1 \\ \hline
	0.9 & 100  & 0.55 & 44   & 0.7 & 1 & 1  & 0.900 & 1 \\ \hline
	0.9 & 100  & 0.55 & 50   & 0.7 & 1 & 1  & 0.900 & 1 \\ \hline
	0.9 & 100  & 0.55 & 98   & 0.7 & 1 & 1  & 0.900 & 1 \\ \hline
	0.9 & 100  & 0.75 & 5    & 0.5 & 1 & 1  & 0.900 & 1 \\ \hline
	0.9 & 100  & 0.75 & 24   & 0.5 & 1 & 1  & 0.900 & 1 \\ \hline
	0.9 & 100  & 0.75 & 44   & 0.5 & 1 & 1  & 0.900 & 1 \\ \hline
	0.9 & 100  & 0.75 & 50   & 0.5 & 1 & 1  & 0.900 & 1 \\ \hline
	0.9 & 100  & 0.75 & 98   & 0.5 & 1 & 1  & 0.900 & 1 \\ \hline
	0.9 & 100  & 0.75 & 5    & 0.7 & 1 & 1  & 0.900 & 1 \\ \hline
	0.9 & 100  & 0.75 & 24   & 0.7 & 1 & 1  & 0.900 & 1 \\ \hline
	0.9 & 100  & 0.75 & 44   & 0.7 & 1 & 1  & 0.900 & 1 \\ \hline
	0.9 & 100  & 0.75 & 50   & 0.7 & 1 & 1  & 0.900 & 1 \\ \hline
	0.9 & 100  & 0.75 & 98   & 0.7 & 1 & 1  & 0.900 & 1 \\ \hline
	0.2 & 1000 & 0.55 & 50   & 0.5 & 0.097 & 0  & 0.790 & 0.016 \\ \hline
	0.2 & 1000 & 0.55 & 240  & 0.5 & 0.516 & 0  & 0.771 & 0.113 \\ \hline
	0.2 & 1000 & 0.55 & 440  & 0.5 & 0.958 & 0  & 0.766 & 0.211 \\ \hline
	0.2 & 1000 & 0.55 & 500  & 0.5 & 1 & 0.065  & 0.844 & 0.195 \\ \hline
	0.2 & 1000 & 0.55 & 980  & 0.5 & 1 & 0.948  & 0.810 & 0.100 \\ \hline
	0.2 & 1000 & 0.55 & 50   & 0.7 & 0.085 & 0  & 0.838 & 0.005 \\ \hline
	0.2 & 1000 & 0.55 & 240  & 0.7 & 0.497 & 0  & 0.904 & 0.078 \\ \hline
	0.2 & 1000 & 0.55 & 440  & 0.7 & 0.936 & 0  & 0.901 & 0.147 \\ \hline
	0.2 & 1000 & 0.55 & 500  & 0.7 & 1 & 0.050  & 0.934 & 0.148 \\ \hline
	0.2 & 1000 & 0.55 & 980  & 0.7 & 1 & 0.939  & 0.811 & 0.022 \\ \hline
	0.2 & 1000 & 0.75 & 50   & 0.5 & 0.214 & 0  & 0.436 & 0.064 \\ \hline
	0.2 & 1000 & 0.75 & 240  & 0.5 & 0.992 & 0  & 0.414 & 0.298 \\ \hline
	0.2 & 1000 & 0.75 & 440  & 0.5 & 1 & 0.223  & 0.938 & 0.189 \\ \hline
	0.2 & 1000 & 0.75 & 500  & 0.5 & 1 & 0.303  & 0.937 & 0.197 \\ \hline
	0.2 & 1000 & 0.75 & 980  & 0.5 & 1 & 0.963  & 0.934 & 0.263 \\ \hline
	0.2 & 1000 & 0.75 & 50   & 0.7 & 0.182 & 0  & 0.782 & 0.039 \\ \hline
	0.2 & 1000 & 0.75 & 240  & 0.7 & 0.930 & 0  & 0.759 & 0.199 \\ \hline
	0.2 & 1000 & 0.75 & 440  & 0.7 & 1 & 0.214  & 0.974 & 0.172 \\ \hline
	0.2 & 1000 & 0.75 & 500  & 0.7 & 1 & 0.294  & 0.973 & 0.178 \\ \hline
	0.2 & 1000 & 0.75 & 980  & 0.7 & 1 & 0.960  & 0.972 & 0.225 \\ \hline
	0.5 & 1000 & 0.55 & 50   & 0.5 & 1 & 0  & 0.500 & 0.500 \\ \hline
	0.5 & 1000 & 0.55 & 240  & 0.5 & 1 & 0  & 0.500 & 0.500 \\ \hline
	0.5 & 1000 & 0.55 & 440  & 0.5 & 1 & 0.006  & 0.596 & 0.462 \\ \hline
	0.5 & 1000 & 0.55 & 500  & 0.5 & 1 & 0.115  & 0.595 & 0.511 \\ \hline
	0.5 & 1000 & 0.55 & 980  & 0.5 & 1 & 0.968  & 0.586 & 0.845 \\ \hline
	0.5 & 1000 & 0.55 & 50   & 0.7 & 0.108 & 0  & 0.738 & 0.039 \\ \hline
	0.5 & 1000 & 0.55 & 240  & 0.7 & 0.533 & 0  & 0.749 & 0.209 \\ \hline
	0.5 & 1000 & 0.55 & 440  & 0.7 & 0.978 & 0  & 0.745 & 0.384 \\ \hline
	0.5 & 1000 & 0.55 & 500  & 0.7 & 1 & 0.078  & 0.831 & 0.379 \\ \hline
	0.5 & 1000 & 0.55 & 980  & 0.7 & 1 & 0.954  & 0.668 & 0.342 \\ \hline
	0.5 & 1000 & 0.75 & 50   & 0.5 & 1 & 0  & 0.500 & 0.500 \\ \hline
	0.5 & 1000 & 0.75 & 240  & 0.5 & 1 & 0  & 0.611 & 0.474 \\ \hline
	0.5 & 1000 & 0.75 & 440  & 0.5 & 1 & 0.243  & 0.840 & 0.477 \\ \hline
	0.5 & 1000 & 0.75 & 500  & 0.5 & 1 & 0.323  & 0.839 & 0.497 \\ \hline
	0.5 & 1000 & 0.75 & 980  & 0.5 & 1 & 0.970  & 0.834 & 0.659 \\ \hline
	0.5 & 1000 & 0.75 & 50   & 0.7 & 1 & 0  & 0.500 & 0.500 \\ \hline
	0.5 & 1000 & 0.75 & 240  & 0.7 & 1 & 0  & 0.611 & 0.474 \\ \hline
	0.5 & 1000 & 0.75 & 440  & 0.7 & 1 & 0.230  & 0.933 & 0.434 \\ \hline
	0.5 & 1000 & 0.75 & 500  & 0.7 & 1 & 0.309  & 0.932 & 0.448 \\ \hline
	0.5 & 1000 & 0.75 & 980  & 0.7 & 1 & 0.966  & 0.929 & 0.565 \\ \hline
	0.7 & 1000 & 0.55 & 50   & 0.5 & 1 & 1  & 0.700 & 1 \\ \hline
	0.7 & 1000 & 0.55 & 240  & 0.5 & 1 & 1  & 0.700 & 1 \\ \hline
	0.7 & 1000 & 0.55 & 440  & 0.5 & 1 & 1  & 0.700 & 1 \\ \hline
	0.7 & 1000 & 0.55 & 500  & 0.5 & 1 & 1  & 0.700 & 1 \\ \hline
	0.7 & 1000 & 0.55 & 980  & 0.5 & 1 & 1  & 0.700 & 1 \\ \hline
	0.7 & 1000 & 0.55 & 50   & 0.7 & 1 & 0  & 0.700 & 0.520 \\ \hline
	0.7 & 1000 & 0.55 & 240  & 0.7 & 1 & 0  & 0.700 & 0.520 \\ \hline
	0.7 & 1000 & 0.55 & 440  & 0.7 & 1 & 0.006  & 0.757 & 0.498 \\ \hline
	0.7 & 1000 & 0.55 & 500  & 0.7 & 1 & 0.115  & 0.757 & 0.547 \\ \hline
	0.7 & 1000 & 0.55 & 980  & 0.7 & 1 & 0.968  & 0.729 & 0.879 \\ \hline
	0.7 & 1000 & 0.75 & 50   & 0.5 & 1 & 0  & 0.700 & 0.600 \\ \hline
	0.7 & 1000 & 0.75 & 240  & 0.5 & 1 & 0  & 0.767 & 0.585 \\ \hline
	0.7 & 1000 & 0.75 & 440  & 0.5 & 1 & 0.268  & 0.770 & 0.677 \\ \hline
	0.7 & 1000 & 0.75 & 500  & 0.5 & 1 & 0.348  & 0.769 & 0.705 \\ \hline
	0.7 & 1000 & 0.75 & 980  & 0.5 & 1 & 0.977  & 0.767 & 0.925 \\ \hline
	0.7 & 1000 & 0.75 & 50   & 0.7 & 1 & 0  & 0.700 & 0.600 \\ \hline
	0.7 & 1000 & 0.75 & 240  & 0.7 & 1 & 0  & 0.767 & 0.585 \\ \hline
	0.7 & 1000 & 0.75 & 440  & 0.7 & 1 & 0.243  & 0.904 & 0.611 \\ \hline
	0.7 & 1000 & 0.75 & 500  & 0.7 & 1 & 0.323  & 0.903 & 0.631 \\ \hline
	0.7 & 1000 & 0.75 & 980  & 0.7 & 1 & 0.970  & 0.900 & 0.792 \\ \hline
	0.9 & 1000 & 0.55 & 50   & 0.5 & 1 & 1  & 0.900 & 1 \\ \hline
	0.9 & 1000 & 0.55 & 240  & 0.5 & 1 & 1  & 0.900 & 1 \\ \hline
	0.9 & 1000 & 0.55 & 440  & 0.5 & 1 & 1  & 0.900 & 1 \\ \hline
	0.9 & 1000 & 0.55 & 500  & 0.5 & 1 & 1  & 0.900 & 1 \\ \hline
	0.9 & 1000 & 0.55 & 980  & 0.5 & 1 & 1  & 0.900 & 1 \\ \hline
	0.9 & 1000 & 0.55 & 50   & 0.7 & 1 & 1  & 0.900 & 1 \\ \hline
	0.9 & 1000 & 0.55 & 240  & 0.7 & 1 & 1  & 0.900 & 1 \\ \hline
	0.9 & 1000 & 0.55 & 440  & 0.7 & 1 & 1  & 0.900 & 1 \\ \hline
	0.9 & 1000 & 0.55 & 500  & 0.7 & 1 & 1  & 0.900 & 1 \\ \hline
	0.9 & 1000 & 0.55 & 980  & 0.7 & 1 & 1  & 0.900 & 1 \\ \hline
	0.9 & 1000 & 0.75 & 50   & 0.5 & 1 & 1  & 0.900 & 1 \\ \hline
	0.9 & 1000 & 0.75 & 240  & 0.5 & 1 & 1  & 0.900 & 1 \\ \hline
	0.9 & 1000 & 0.75 & 440  & 0.5 & 1 & 1  & 0.900 & 1 \\ \hline
	0.9 & 1000 & 0.75 & 500  & 0.5 & 1 & 1  & 0.900 & 1 \\ \hline
	0.9 & 1000 & 0.75 & 980  & 0.5 & 1 & 1  & 0.900 & 1 \\ \hline
	0.9 & 1000 & 0.75 & 50   & 0.7 & 1 & 1  & 0.900 & 1 \\ \hline
	0.9 & 1000 & 0.75 & 240  & 0.7 & 1 & 1  & 0.900 & 1 \\ \hline
	0.9 & 1000 & 0.75 & 440  & 0.7 & 1 & 1  & 0.900 & 1 \\ \hline
	0.9 & 1000 & 0.75 & 500  & 0.7 & 1 & 1  & 0.900 & 1 \\ \hline
	0.9 & 1000 & 0.75 & 980  & 0.7 & 1 & 1  & 0.900 & 1 \\ \hline
	0.2 & 5000 & 0.55 & 250  & 0.5 & 0.107 & 0  & 0.773 & 0.024 \\ \hline
	0.2 & 5000 & 0.55 & 1200 & 0.5 & 0.526 & 0  & 0.763 & 0.116 \\ \hline
	0.2 & 5000 & 0.55 & 2200 & 0.5 & 0.969 & 0  & 0.760 & 0.213 \\ \hline
	0.2 & 5000 & 0.55 & 2500 & 0.5 & 1 & 0.080  & 0.840 & 0.198 \\ \hline
	0.2 & 5000 & 0.55 & 4900 & 0.5 & 1 & 0.500  & 0.800 & 0 \\ \hline
	0.2 & 5000 & 0.55 & 250  & 0.7 & 0.103 & 0  & 0.903 & 0.016 \\ \hline
	0.2 & 5000 & 0.55 & 1200 & 0.7 & 0.518 & 0  & 0.899 & 0.081 \\ \hline
	0.2 & 5000 & 0.55 & 2200 & 0.7 & 0.960 & 0  & 0.898 & 0.151 \\ \hline
	0.2 & 5000 & 0.55 & 2500 & 0.7 & 1 & 0.073  & 0.932 & 0.151 \\ \hline
	0.2 & 5000 & 0.55 & 4900 & 0.7 & 1 & 0.500  & 0.800 & 0 \\ \hline
	0.2 & 5000 & 0.75 & 250  & 0.5 & 0.207 & 0  & 0.416 & 0.062 \\ \hline
	0.2 & 5000 & 0.75 & 1200 & 0.5 & 0.975 & 0  & 0.406 & 0.293 \\ \hline
	0.2 & 5000 & 0.75 & 2200 & 0.5 & 1 & 0.240  & 0.935 & 0.191 \\ \hline
	0.2 & 5000 & 0.75 & 2500 & 0.5 & 1 & 0.320  & 0.935 & 0.199 \\ \hline
	0.2 & 5000 & 0.75 & 4900 & 0.5 & 1 & 0.969  & 0.934 & 0.264 \\ \hline
	0.2 & 5000 & 0.75 & 250  & 0.7 & 0.193 & 0  & 0.760 & 0.041 \\ \hline
	0.2 & 5000 & 0.75 & 1200 & 0.7 & 0.948 & 0  & 0.750 & 0.203 \\ \hline
	0.2 & 5000 & 0.75 & 2200 & 0.7 & 1 & 0.236  & 0.972 & 0.174 \\ \hline
	0.2 & 5000 & 0.75 & 2500 & 0.7 & 1 & 0.316  & 0.972 & 0.180 \\ \hline
	0.2 & 5000 & 0.75 & 4900 & 0.7 & 1 & 0.968  & 0.972 & 0.226 \\ \hline
	0.5 & 5000 & 0.55 & 250  & 0.5 & 1 & 0  & 0.500 & 0.500 \\ \hline
	0.5 & 5000 & 0.55 & 1200 & 0.5 & 1 & 0  & 0.500 & 0.500 \\ \hline
	0.5 & 5000 & 0.55 & 2200 & 0.5 & 1 & 0  & 0.538 & 0.484 \\ \hline
	0.5 & 5000 & 0.55 & 2500 & 0.5 & 1 & 0.102  & 0.593 & 0.505 \\ \hline
	0.5 & 5000 & 0.55 & 4900 & 0.5 & 1 & 0.500  & 0.500 & 0 \\ \hline
	0.5 & 5000 & 0.55 & 250  & 0.7 & 0.111 & 0  & 0.750 & 0.044 \\ \hline
	0.5 & 5000 & 0.55 & 1200 & 0.7 & 0.534 & 0  & 0.743 & 0.210 \\ \hline
	0.5 & 5000 & 0.55 & 2200 & 0.7 & 0.978 & 0  & 0.741 & 0.384 \\ \hline
	0.5 & 5000 & 0.55 & 2500 & 0.7 & 1 & 0.086  & 0.827 & 0.381 \\ \hline
	0.5 & 5000 & 0.55 & 4900 & 0.7 & 1 & 0.500  & 0.500 & 0 \\ \hline
	0.5 & 5000 & 0.75 & 250  & 0.5 & 1 & 0  & 0.500 & 0.500 \\ \hline
	0.5 & 5000 & 0.75 & 1200 & 0.5 & 1 & 0  & 0.525 & 0.494 \\ \hline
	0.5 & 5000 & 0.75 & 2200 & 0.5 & 1 & 0.249  & 0.836 & 0.479 \\ \hline
	0.5 & 5000 & 0.75 & 2500 & 0.5 & 1 & 0.329  & 0.836 & 0.499 \\ \hline
	0.5 & 5000 & 0.75 & 4900 & 0.5 & 1 & 0.972  & 0.834 & 0.660 \\ \hline
	0.5 & 5000 & 0.75 & 250  & 0.7 & 1 & 0  & 0.500 & 0.500 \\ \hline
	0.5 & 5000 & 0.75 & 1200 & 0.7 & 1 & 0  & 0.525 & 0.494 \\ \hline
	0.5 & 5000 & 0.75 & 2200 & 0.7 & 1 & 0.243  & 0.930 & 0.436 \\ \hline
	0.5 & 5000 & 0.75 & 2500 & 0.7 & 1 & 0.323  & 0.930 & 0.451 \\ \hline
	0.5 & 5000 & 0.75 & 4900 & 0.7 & 1 & 0.970  & 0.929 & 0.566 \\ \hline
	0.7 & 5000 & 0.55 & 250  & 0.5 & 1 & 1  & 0.700 & 1 \\ \hline
	0.7 & 5000 & 0.55 & 1200 & 0.5 & 1 & 1  & 0.700 & 1 \\ \hline
	0.7 & 5000 & 0.55 & 2200 & 0.5 & 1 & 1  & 0.700 & 1 \\ \hline
	0.7 & 5000 & 0.55 & 2500 & 0.5 & 1 & 1  & 0.700 & 1 \\ \hline
	0.7 & 5000 & 0.55 & 4900 & 0.5 & 1 & 0.500  & 0.300 & 0 \\ \hline
	0.7 & 5000 & 0.55 & 250  & 0.7 & 1 & 0  & 0.700 & 0.520 \\ \hline
	0.7 & 5000 & 0.55 & 1200 & 0.7 & 1 & 0  & 0.700 & 0.520 \\ \hline
	0.7 & 5000 & 0.55 & 2200 & 0.7 & 1 & 0  & 0.723 & 0.510 \\ \hline
	0.7 & 5000 & 0.55 & 2500 & 0.7 & 1 & 0.102  & 0.756 & 0.541 \\ \hline
	0.7 & 5000 & 0.55 & 4900 & 0.7 & 1 & 0.500  & 0.300 & 0 \\ \hline
	0.7 & 5000 & 0.75 & 250  & 0.5 & 1 & 0  & 0.700 & 0.600 \\ \hline
	0.7 & 5000 & 0.75 & 1200 & 0.5 & 1 & 0  & 0.715 & 0.597 \\ \hline
	0.7 & 5000 & 0.75 & 2200 & 0.5 & 1 & 0.260  & 0.768 & 0.674 \\ \hline
	0.7 & 5000 & 0.75 & 2500 & 0.5 & 1 & 0.340  & 0.768 & 0.702 \\ \hline
	0.7 & 5000 & 0.75 & 4900 & 0.5 & 1 & 0.975  & 0.767 & 0.925 \\ \hline
	0.7 & 5000 & 0.75 & 250  & 0.7 & 1 & 0  & 0.700 & 0.600 \\ \hline
	0.7 & 5000 & 0.75 & 1200 & 0.7 & 1 & 0  & 0.715 & 0.597 \\ \hline
	0.7 & 5000 & 0.75 & 2200 & 0.7 & 1 & 0.249  & 0.902 & 0.612 \\ \hline
	0.7 & 5000 & 0.75 & 2500 & 0.7 & 1 & 0.329  & 0.902 & 0.632 \\ \hline
	0.7 & 5000 & 0.75 & 4900 & 0.7 & 1 & 0.972  & 0.900 & 0.793 \\ \hline
	0.9 & 5000 & 0.55 & 250  & 0.5 & 1 & 1  & 0.900 & 1 \\ \hline
	0.9 & 5000 & 0.55 & 1200 & 0.5 & 1 & 1  & 0.900 & 1 \\ \hline
	0.9 & 5000 & 0.55 & 2200 & 0.5 & 1 & 1  & 0.900 & 1 \\ \hline
	0.9 & 5000 & 0.55 & 2500 & 0.5 & 1 & 1  & 0.900 & 1 \\ \hline
	0.9 & 5000 & 0.55 & 4900 & 0.5 & 1 & 0.500  & 0.100 & 0 \\ \hline
	0.9 & 5000 & 0.55 & 250  & 0.7 & 1 & 1  & 0.900 & 1 \\ \hline
	0.9 & 5000 & 0.55 & 1200 & 0.7 & 1 & 1  & 0.900 & 1 \\ \hline
	0.9 & 5000 & 0.55 & 2200 & 0.7 & 1 & 1  & 0.900 & 1 \\ \hline
	0.9 & 5000 & 0.55 & 2500 & 0.7 & 1 & 1  & 0.900 & 1 \\ \hline
	0.9 & 5000 & 0.55 & 4900 & 0.7 & 1 & 0.500  & 0.100 & 0 \\ \hline
	0.9 & 5000 & 0.75 & 250  & 0.5 & 1 & 1  & 0.900 & 1 \\ \hline
	0.9 & 5000 & 0.75 & 1200 & 0.5 & 1 & 1  & 0.900 & 1 \\ \hline
	0.9 & 5000 & 0.75 & 2200 & 0.5 & 1 & 1  & 0.900 & 1 \\ \hline
	0.9 & 5000 & 0.75 & 2500 & 0.5 & 1 & 1  & 0.900 & 1 \\ \hline
	0.9 & 5000 & 0.75 & 4900 & 0.5 & 1 & 1  & 0.900 & 1 \\ \hline
	0.9 & 5000 & 0.75 & 250  & 0.7 & 1 & 1  & 0.900 & 1 \\ \hline
	0.9 & 5000 & 0.75 & 1200 & 0.7 & 1 & 1  & 0.900 & 1 \\ \hline
	0.9 & 5000 & 0.75 & 2200 & 0.7 & 1 & 1  & 0.900 & 1 \\ \hline
	0.9 & 5000 & 0.75 & 2500 & 0.7 & 1 & 1  & 0.900 & 1 \\ \hline
	0.9 & 5000 & 0.75 & 4900 & 0.7 & 1 & 1  & 0.900 & 1 \\ \hline
\end{longtable}

\end{document}